\documentclass[USenglish,oneside,twocolumn]{article}
\usepackage[utf8]{inputenc}
\usepackage[big]{dgruyter_NEW}
\usepackage{array}
\usepackage{float}
\usepackage{calc}
\usepackage{booktabs}
\usepackage{amsmath}
\usepackage{amsthm}
\usepackage{amssymb}
\usepackage{graphicx}
\usepackage{textcomp}
\usepackage{url}

\providecommand{\tabularnewline}{\\}
\floatstyle{ruled}
\newfloat{algorithm}{tbp}{loa}
\providecommand{\algorithmname}{Algorithm}
\floatname{algorithm}{\protect\algorithmname}

\makeatletter
\def\plist@algorithm{\algorithmname\space}
\makeatother

\newcommand{\noun}[1]{\textsc{#1}}
\theoremstyle{definition}
\newtheorem{defn}{\protect\definitionname}
\theoremstyle{remark}
\newtheorem{rem}{\protect\remarkname}
\theoremstyle{plain}
\newtheorem{thm}{\protect\theoremname}
\theoremstyle{plain}
\newtheorem{lem}{\protect\lemmaname}

\providecommand{\definitionname}{Definition}
\providecommand{\lemmaname}{Lemma}
\providecommand{\remarkname}{Remark}
\providecommand{\theoremname}{Theorem}

\usepackage{color}
\newcommand{\chang}{\textcolor{black}}

\cclogo{\includegraphics{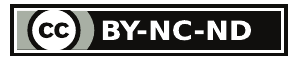}}

\begin{document}

  \author*[1]{Thee Chanyaswad}

  \author[2]{Changchang Liu}

  \author[3]{Prateek Mittal}


  \affil[1]{Princeton University, E-mail: tc7@princeton.edu (Currently at KBTG Machine Learning Team, Thailand, E-mail: theerachai.c@kbtg.tech)}

  \affil[2]{Princeton University, E-mail: cl12@princeton.edu  (Currently at IBM T. J. Watson Research Center, E-mail: changchang.liu33@ibm.com)}

  \affil[3]{Princeton University, E-mail: pmittal@princeton.edu}


  \title{\huge RON-Gauss: Enhancing Utility in Non-Interactive
Private Data Release}

  \runningtitle{RON-Gauss: Enhancing Utility in Non-Interactive Private Data Release}

  \begin{abstract}
{
A key challenge facing the design of differential privacy in the non-interactive
setting is to maintain the utility of the released data. To overcome
this challenge, we utilize the \emph{Diaconis-Freedman-Meckes (DFM)
effect}, which states that most projections of high-dimensional data
are nearly Gaussian. Hence, we propose the \emph{RON-Gauss} model
that leverages the novel combination of dimensionality reduction via
random orthonormal (RON) projection and the Gaussian generative model
for synthesizing differentially-private data. We analyze how RON-Gauss
benefits from the DFM effect, and present multiple algorithms for
a range of machine learning applications, including both unsupervised
and supervised learning. Furthermore, we rigorously prove that (a)
our algorithms satisfy the strong $\epsilon$-differential privacy
guarantee, and (b) RON projection can lower the level of perturbation
required for differential privacy. Finally, we illustrate the effectiveness
of RON-Gauss under three common machine learning applications \textendash{}
clustering, classification, and regression \textendash{} on three
large real-world datasets. Our empirical results show that (a) RON-Gauss
outperforms previous approaches by up to an order of magnitude, and (b)
loss in utility compared to the non-private real data is small. Thus,
RON-Gauss can serve as a key enabler for real-world deployment of
privacy-preserving data release.
}
\end{abstract}

  \journalname{Proceedings on Privacy Enhancing Technologies}
\DOI{Editor to enter DOI}
  \startpage{1}
  \received{..}
  \revised{..}
  \accepted{..}

  \journalyear{..}
  \journalvolume{..}
  \journalissue{..}

\maketitle

\section{Introduction}

In an era of big data and machine learning, our digital society is
generating a considerable amount of personal data at every moment.
These data can be sensitive, and as a result, significant privacy
concerns arise. Even with the use of anonymization mechanisms, privacy
leakage can still occur, as exemplified by Narayanan et al. \cite{RefWorks:181},
Calandrino et al. \cite{RefWorks:182}, Barbaro and Zeller \cite{RefWorks:183},
Haeberlen et al. \cite{RefWorks:260}, and Backes et al. \cite{RefWorks:262}.
These privacy leaks have motivated the design of formal privacy analysis.
To this end, \emph{differential privacy (DP)} has become the gold
standard for a rigorous privacy guarantee \cite{RefWorks:151,RefWorks:185,RefWorks:186,RefWorks:194,RefWorks:195}.
Many mechanisms have been proposed to comply with differential privacy \cite{RefWorks:193,RefWorks:194,RefWorks:188,RefWorks:185,RefWorks:195,RefWorks:192,RefWorks:220,RefWorks:221,RefWorks:222,RefWorks:224,RefWorks:246},
and various implementations of differentially-private systems have
been presented in the literature \cite{RefWorks:254,RefWorks:255,RefWorks:256,RefWorks:257,RefWorks:258,RefWorks:261,RefWorks:265,RefWorks:379,blocki2016differentially,RefWorks:449}.

There are two settings under differential privacy \textendash{} interactive
and non-interactive \cite{RefWorks:151}. Among the two, the non-interactive
setting has traditionally been more challenging to implement due to
the fact that the perturbation required is often too high for the
published data to be truly useful \cite{RefWorks:381,RefWorks:336,RefWorks:527,RefWorks:530}.
However, this setting is still attractive, as there are incentives
for the data collectors to release the data in order to seek outside
expertise, e.g. the Netflix prize \cite{RefWorks:377}, and OpenSNP
\cite{RefWorks:378}. Concurrently, there are incentives for researchers
to obtain the data in their entirety, as existing software frameworks
for data analytics could be directly used \cite{RefWorks:336,RefWorks:372}.
Particularly, in the current era when machine learning has become
the ubiquitous tool in data analysis, non-interactive data release
would allow virtually instant compatibility with existing 
learning algorithms. For these reasons, we aim to design a
non-interactive differentially-private (DP) data release system.

In this work, we draw inspiration from the \emph{Diaconis-Freedman-Meckes
(DFM) effect} \cite{RefWorks:439}, which shows that, under suitable
conditions, most projections of high-dimensional data are nearly
\emph{Gaussian}. This effect suggests that, although finding an accurate
model for a high-dimensional dataset is generally hard \cite{RefWorks:448}\cite[chapter 7]{RefWorks:376},
its projection onto a low-dimensional space may be modeled well by
the Gaussian model. With respect to the application of non-interactive
DP data release, this is particularly important because, in DP, simple
statistics can generally be privately learned \emph{accurately} \cite{RefWorks:477,RefWorks:478},
while privately learning the database accurately is generally much
more difficult \cite{RefWorks:527,RefWorks:530,RefWorks:479}.

To apply the DFM effect to the non-interactive DP data release, we
combine two previously non-intersecting methods \textendash{} \emph{dimensionality
reduction (DR)} \cite{RefWorks:313,RefWorks:314,RefWorks:337,RefWorks:338,RefWorks:339,RefWorks:224}
and \emph{parametric generative model} \cite{RefWorks:238,RefWorks:187,RefWorks:333,RefWorks:334,RefWorks:335,RefWorks:336,RefWorks:384}.
Although each method has independently been explored for the application
of non-interactive private data release, without properly combining
the two, the DFM effect has not been fully utilized. As we show in
this work, \emph{combining the two} to utilize the DFM effect can
lead to significant gain in the utility of the released data. Specifically,
we closely investigate the DFM theorem by Meckes \cite{RefWorks:439}
and propose the \emph{RON-Gauss} model
for non-interactive private data release, which combines two techniques
\textendash{} \emph{random orthonormal (RON) projection} and the \emph{Gaussian
generative model}. The first component is the DR technique used for
two purposes: reducing the sensitivity (similar to previous works \cite{RefWorks:313,RefWorks:314,RefWorks:337,RefWorks:338,RefWorks:339,RefWorks:224}) and triggering
the DFM effect (which is a first, to the
best of our knowledge). The second component is the parametric model used
to capture the Gaussian nature of the projection and to allow an accurate DP data modeling.

We present three algorithms for RON-Gauss
that can be applied to a wide range of machine learning applications,
including both \emph{unsupervised} and \emph{supervised} learning.
The supervised learning application, in particular, provides an additional
challenge on the conservation of the training label through the sanitization
process. Unlike many previous works, RON-Gauss ensures the integrity
of the training label of the sanitized data. We rigorously prove
that all of our three algorithms preserve the strong $\epsilon$-differential
privacy guarantee. Moreover, to show the general applicability of
our idea, we extend the framework to employ the \emph{Gaussian Mixture
Model (GMM)} \cite{RefWorks:225,RefWorks:51}.

Finally, we evaluate our approach on three large real-world datasets
under three common machine learning applications under the non-interactive setting of DP \textendash{} clustering,
classification, and regression. The non-interactive setting is attractive for these applications since it allows multiple data-analytic algorithms to be run on the released DP-data without requiring additional privacy budget like the interactive setting. We demonstrate that our method can
significantly improve the utility performance by up to an order of magnitude
for a fixed privacy budget, when compared to four prior methods. More
importantly, our method has small loss in utility when compared to
the performance of the non-private real data.

We summarize our contribution as follows. 
\begin{itemize}
\item We exploit the DFM effect for utility enhancement of differential privacy in the non-interactive setting.
\item We propose an approach consisting of random orthonormal projection
and the Gaussian generative model (\emph{RON-Gauss}) for non-interactive
DP data release. We also extend this model to the Gaussian Mixture
Model (GMM).
\item We present three algorithms to implement RON-Gauss \chang{that are suitable for} both
the \emph{unsupervised} and \emph{supervised} machine learning tasks.
\item We rigorously prove that our RON-Gauss algorithms satisfy the strong
$\epsilon$-differential privacy.
\item We evaluate our method on three real-world datasets on three machine learning applications under the non-interactive DP setting \textendash{} \emph{clustering}, \emph{classification},
and \emph{regression}. The experimental results show that, when compared
to previous methods, our method can considerably enhance the utility
performance by up to an order of magnitude for a fixed privacy budget. Finally,
compared to the non-private baseline of using real data, RON-Gauss
incurs only a small loss in utility across all three machine learning
tasks.
\end{itemize}
\chang{Roadmap: We discuss prior works in Section \ref{sec:Prior-Works}, and present the background components of our approach, including details of the DFM effect, in Section \ref{sec:preliminaries}. Then, we present the proposed RON-Gauss model \textendash{} along with its theoretical analysis, algorithms for both supervised and unsupervised learning, and the privacy proofs \textendash{} in Section \ref{sec:RON-Gauss}. Finally, we present  experimental results showing the strength of RON-Gauss in Section~\ref{sec:Experiments}, and 
the discussion in Section \ref{sec:discussion}.
}

\vspace{-1em}
\section{Prior Works}

\label{sec:Prior-Works}

Our work focuses on non-interactive differentially-private (DP)
data release. Since our method involves dimensionality reduction and a generative model, we discuss the relevant works under these
frameworks.

\vspace{-1em}
\subsection{Generative Models for Differential Privacy}

\label{subsec:generative_models} The use of generative models for non-interactive DP data release can be classified
into two groups according to Bowen and Liu \cite{RefWorks:370}: \emph{non-parametric
generative models}, and \emph{parametric generative models}.

\vspace{-1em}
\subsubsection{Non-Parametric Generative Models}

Primarily, these models utilize the differential privacy guarantee
of the Exponential mechanism \cite{RefWorks:192}, which defines a
distribution to synthesize the data based on the input database and
the pre-defined quality function. Various methods \textendash{} both
application-specific and application-independent \textendash{} have
been proposed \cite{RefWorks:174,RefWorks:192,RefWorks:187,RefWorks:179,RefWorks:237,RefWorks:221,RefWorks:382,RefWorks:383,RefWorks:450,RefWorks:479,RefWorks:480,RefWorks:481,RefWorks:512,RefWorks:449}.
Our approach contrasts these works in two ways. First, we consider
a parametric generative model, and, second, we augment our model with
dimensionality reduction to trigger the DFM effect. We will compare our method
to this class of model by implementing the non-parametric generative
model based on Algorithm 1 in \cite{RefWorks:174}.

\vspace{-1em}
\subsubsection{Parametric Generative Models}

Our method of using the \emph{Gaussian generative model}, as well as the \emph{Gaussian
Mixture Model}, falls into this category. We aim at building a system
that can be applied to various applications and data-types,
i.e. \emph{application-independent}. However, many previous
works on non-interactive DP data release are application-specific
or limited by the data-types they are compatible with. Thus, we discuss
these two types separately.

\vspace{-1em}
\paragraph{Application-Specific}

These models are designed for specific applications or data-types.
For example, the works by Sala et al. \cite{RefWorks:238} and by
Proserpio et al. \cite{RefWorks:187} are for graph analysis, the
system by Ororbia et al. \cite{RefWorks:334} is for plaintext statistics,
the analysis by Machanavajjhala et al. \cite{RefWorks:333} is for
commuting pattern analysis, the Binomial-Beta model by McClure and
Reiter \cite{RefWorks:371} and Bayesian-network by Zhang et al. \cite{RefWorks:335}
are for binary data, and the LDA model by Jiang et al. \cite{RefWorks:339}
is for binary classification. In contrast, in this work, we aim at designing an application-independent generative model.

\vspace{-1em}
\paragraph{Application-Independent}

These generative models are less common, possibly due to the fact
that releasing data for general analytics often requires a high level
of perturbation that impacts data utility. Bindschaedler et al. \cite{RefWorks:336}
design a system for \emph{plausible deniability}, which can be extended
to $(\epsilon,\delta)$-differential privacy. Acs et al. \cite{RefWorks:476}
design a system based on two steps \textendash{} kernel K-means clustering
and generative neural networks \textendash{} to similarly provide
$(\epsilon,\delta)$-differential privacy. In contrast, our work aims
at providing the strictly stronger $\epsilon$-differential privacy.
Another previous method is MODIPS by Liu \cite{RefWorks:372}, which
applies statistical models based on the concept of sufficient statistics
to capture the distribution of the data, and then synthesizes the
data from the differentially-private models. This general idea is, in fact, closely related to the Gaussian generative
model employed in this work. However, the important distinction is that
MODIPS is not accompanied by dimensionality reduction \textendash{} a step which
will be shown to enhance the utility of released
data via the DFM effect. For comparison, we implement MODIPS in our experiments and show the improvement achievable by RON-Gauss.

\vspace{-0.5em}
\subsection{Dimensionality Reduction and Differential Privacy}

\label{subsed:dim_for_improve_utility} 
Traditionally, data partition
and aggregation \cite{RefWorks:241,RefWorks:242,RefWorks:243,RefWorks:244,RefWorks:245,RefWorks:246,RefWorks:247,RefWorks:248,RefWorks:193,RefWorks:253,RefWorks:187,RefWorks:221,RefWorks:236,RefWorks:237}
have been applied to enhance data utility in differential
privacy. In contrast, our work utilizes the DFM effect for the utility enhancement, of which an important component is \emph{dimensionality reduction (DR)} using the RON projection.  We present previous works pertaining to the use of DR in DP here. However, although previous works have explored the use of DR to directly provide DP or to reduce the sensitivity of the query, our work, in contrast, uses DR primarily to trigger the DFM effect for enhancing data utility in the non-interactive setting.

\vspace{-1em}
\subsubsection{Random Projection}

For suitable query functions, random projection
has been shown to preserve differential privacy \cite{RefWorks:313,RefWorks:398,RefWorks:399}. Alternatively, random projection has also been used to enhance the utility of differential privacy. Multiple types
of random projections have been used with the identity query for non-interactive
DP data release for both purposes. For example, Blocki et al. \cite{RefWorks:313},
Kenthapadi et al. \cite{RefWorks:314}, Zhou et al. \cite{RefWorks:338},
and Xu et al. \cite{RefWorks:400} use a random matrix whose entries
are i.i.d. Gaussian, whereas Li et al. \cite{RefWorks:337} use i.i.d. Bernoulli
entries. However, there are three main
contrasts to our work. (1) While Blocki et al. \cite{RefWorks:313}
use random projection to preserve differential privacy, we use random projection\emph{
to enhance utility} via the DFM effect. (2) Instead of i.i.d. Gaussian or Bernoulli
entries, we use \emph{random
orthonormal (RON) projection} to ensure the DFM effect as proved by
Meckes \cite{RefWorks:439}. (3) \emph{As opposed to our approach,
none of the previous random-projection methods couples DR with a generative
model.} We will experimentally compare our work with the method by Li et al. \cite{RefWorks:337}, and show that, by exploiting
the Gaussian phenomenon via the DFM effect, we achieve significant
utility gain.

\vspace{-1em}
\subsubsection{Other Dimensionality Reduction Methods}

Other DR methods have also been used with the identity query to enhance
data utility including PCA \cite{RefWorks:339}, wavelet transform
\cite{RefWorks:224}, and lossy Fourier transform \cite{RefWorks:219}.
In contrast, our DR is coupled with a generative model, rather than
used with the identity query. We experimentally compare our
work with the PCA method by Jiang et al. \cite{RefWorks:339} and
show that our use of the generative model yields significant improvement.

\section{Preliminaries}
\label{sec:preliminaries}
In this section, we discuss important background concepts related
to our work.

\vspace{-1em}
\subsection{Database Notation}

\label{subsec:notation} 

We refer to the database as the \emph{dataset}, which contains $n$
records (samples), each with $m$ real-valued attributes (features)
\textendash{} although, our approach is also compatible with categorical
features since they can be converted to real values with encoding
techniques \cite{RefWorks:226}. With this setup, the dataset can
be represented by the data matrix $\mathbf{X}\in\mathbb{R}^{m\times n}$,
whose column vectors $\mathbf{x}_{j}$ are the samples, with $x_{j}(i)$
refers to the $i^{th}$ feature. Finally, random variables are denoted
by a regular capital letter, e.g. $Z$, which may refer to a random
scalar, vector, or matrix. The reference will be clear from the context.

\vspace{-1em}
\subsection{Differential Privacy (DP)}

Differential privacy (DP) protects against the inference of the participation
of a sample in the dataset as follows.
\begin{defn}[$\epsilon$-DP]
\label{def:differential_privacy}A mechanism $\mathcal{A}$ on a query
function $f(\cdot)$ preserves $\epsilon$- differential privacy if
for all neighboring pairs $\mathbf{X},\mathbf{X}'\in\mathbb{R}^{m\times n}$
which differ in a single record and for all possible measurable outputs
$\mathbf{S}\subseteq\mathcal{R}$, 
\[
\frac{\Pr[\mathcal{A}(f(\mathbf{X}))\in\mathbf{S}]}{\Pr[\mathcal{A}(f(\mathbf{X}'))\in\mathbf{S}]}\leq\exp(\epsilon).
\]
\end{defn}
\begin{rem}
\label{rem:eps-delt-DP}There is also the $(\epsilon,\delta)$-differential
privacy ($(\epsilon,\delta)$-DP) \cite{RefWorks:186,RefWorks:259}, which is a relaxation of
this definition. However, this work focuses primarily on
the stronger $\epsilon$-DP.
\end{rem}
Our approach employs the Laplace mechanism, which uses the
notion of $L_{1}$-sensitivity. \chang{For a general query function whose output can be a $p\times q$ matrix, the $L_{1}$-sensitivity is defined as follows.}
\begin{defn}
The $L_{1}$-sensitivity of a query function $f$: $\mathbb{R}^{m\times n}\rightarrow\mathbb{R}^{p\times q}$
for all neighboring datasets $\mathbf{X},\mathbf{X}'\in\mathbb{R}^{m\times n}$
which differ in a single sample is 
\[
S(f)=\sup_{\mathbf{X},\mathbf{X}'}\left\Vert f(\mathbf{X})-f(\mathbf{X}')\right\Vert _{1}.
\]
\end{defn}
\begin{rem}
\label{rem:neighboring_notion}In DP, the notion of neighboring datasets
$\mathbf{X},\mathbf{X}'$ can be considered in two related ways. The
first is the \emph{unbounded} notion when one record is removed or
added. The second is the \emph{bounded} notion when values of one
record vary. The main difference is that the latter assumes the size
of the dataset $n$ is publicly known, while the former assumes it
to be private. However, the two concepts are closely related and a
mechanism that satisfies one can also satisfy the other with a small
cost (cf. \cite{RefWorks:174}). In the following analysis, we adopt
the latter notion for clarity and mathematical simplicity.
\end{rem}
The main tool for DP guarantee in this work is the Laplace mechanism,
which is recited as follows \cite{RefWorks:195,RefWorks:220}.
\begin{thm}
\label{thm:laplace_matrix}For a query function $f$: $\mathbb{R}^{m\times n}\rightarrow\mathbb{R}^{p\times q}$
with the $L_{1}$-sensitivity $S(f)$, the following mechanism preserves
$\epsilon$-differential privacy: 
\[
San(f(\mathbf{X}))=f(\mathbf{X})+Z,
\]
where $Z\in\mathbb{R}^{p\times q}$ with $z_{j}(i)$ drawn i.i.d.
from the Laplace distribution $Lap(S(f)/\epsilon)$. 
\end{thm}

\vspace{-2em}
\subsection{Gaussian Generative Model}\label{subsec:gaussian_model}
Gaussian generative model synthesizes
the data from the Gaussian distribution $\mathcal{N}(\boldsymbol{\mu},\boldsymbol{\Sigma})$,
which is parameterized by the mean $\boldsymbol{\mu}\in\mathbb{R}^{m}$,
and the covariance $\boldsymbol{\Sigma}\in\mathbb{R}^{m\times m}$
\cite{RefWorks:51,RefWorks:475}. Formally, the Gaussian generative
model has the following density function: 
\begin{equation}
f(\mathbf{x})=\frac{1}{\sqrt{(2\pi)^{m}\det(\boldsymbol{\Sigma})}}\exp(-\frac{1}{2}(\mathbf{x}-\boldsymbol{\mu})^{T}\boldsymbol{\Sigma}^{-1}(\mathbf{x}-\boldsymbol{\mu})).\label{eq:guassian_pdf}
\end{equation}
Hence, to obtain the Gaussian generative model for our application,
we only need to estimate its mean and covariance. This
reduces the difficult problem of data modeling into the much simpler
one of statistical estimation. This is particularly important in DP
since it has been shown that simple statistics of the database
can be privately learned accurately \cite{RefWorks:477,RefWorks:478}.
In addition, this model is supported by the following rationales.
\begin{itemize}
\item It is supported by the Diaconis-Freedman-Meckes (DFM) effect \cite{RefWorks:439},
which may be viewed as an analog of the Central Limit Theorem (CLT)
\cite{RefWorks:342,RefWorks:343} in the feature space. This effect
will be discussed in detail in Section \ref{subsec:Diaconis-Freedman-Effect}.
\item It is simple to use and well-understood. Sampling from it is straightforward,
since there exist multiple available packages, e.g. \cite{RefWorks:227,RefWorks:228,RefWorks:229}. 

\item Various methods in data analysis have both directly and indirectly
utilized the Gaussian model, e.g. linear/quadratic discriminant analysis
\cite{RefWorks:292,RefWorks:482}, PCA \cite{RefWorks:442}\cite[chapter 12]{RefWorks:51},
Gaussian Bayes net \cite[chapter 10]{RefWorks:51}, Gaussian Markov
random field \cite{RefWorks:484}, Restricted Boltzmann machine (RBM)
network \cite{RefWorks:484}, Radial Basis Function (RBF) network
\cite{RefWorks:486}, factor analysis \cite[chapter 12]{RefWorks:51},
and SVM \cite{RefWorks:483}.
\end{itemize}
In spite of these advantages, we acknowledge that there is possibly
no single parametric model that can universally capture all possible
datasets, and generalizing our approach to non-parametric generative
models is an interesting future work.

\vspace{-1em}
\subsection{Diaconis-Freedman-Meckes (DFM) \label{subsec:Diaconis-Freedman-Effect}}

Intuitively, the Diaconis-Freedman-Meckes (DFM) effect \textendash{}
initially proved by Diaconis and Freedman in 1984 \cite{RefWorks:435}
\textendash{} states that \emph{"under suitable conditions,
most projections are approximately Gaussian". }Later,
the precise statement of the conditions and the appropriate projections
has been proved by Meckes \cite{RefWorks:439,RefWorks:513}, and the
phenomenon has been substantiated both theoretically \cite{RefWorks:436}
and empirically \cite{RefWorks:437,RefWorks:445}. This work
considers the theorem by Meckes \cite{RefWorks:439} as follows.
\begin{table}
\begin{centering}
\noindent\fbox{\begin{minipage}[t]{1\columnwidth - 2\fboxsep - 2\fboxrule}%
Let the data $X\in\mathbb{R}^{m}$ be drawn from a distribution $\mathcal{X}$,
which satisfies:
\begin{align*}
\mathbb{E}\left[\left\Vert X\right\Vert ^{2}\right] & =\sigma^{2}m,\\
\sup_{\mathbf{v}\in\mathbb{S}^{m-1}}\mathbb{E}\left\langle \mathbf{v},X\right\rangle ^{2} & \leq1,\\
\mathbb{E}\left[\left|\left\Vert X\right\Vert ^{2}\sigma^{-2}-m\right|\right] & \leq c\sqrt{m}.
\end{align*}
\end{minipage}}
\par\end{centering}
\caption{The regularity conditions for the DFM effect. $\sigma$ is the variance
defined in Theorem \ref{thm:proj_is_gauss}, $c>0$ is a constant,
and $\mathbb{S}^{m-1}$ is the topological sphere (cf. \cite{RefWorks:524}).
\label{tab:The-regularity-conditions}}
\vspace{-1em}
\end{table}
\begin{thm}[{DFM effect \cite[Corollary 4]{RefWorks:439}}]
\label{thm:proj_is_gauss}Let $\mathbf{W}\in\mathbb{R}^{m\times p}$
be a random projection matrix with orthonormal columns, $X\in\mathbb{R}^{m}$
be data drawn i.i.d. from an unknown distribution $\mathcal{X}$,
which satisfies the regularity conditions in Table \ref{tab:The-regularity-conditions},
and let $\widetilde{X}=\mathbf{W}^{T}X\in\mathbb{R}^{p}$ be the projection
of $X$ via $\mathbf{W}$, which has the distribution $\widetilde{\mathcal{X}}$.
Then, for $p\ll m$, with high probability,
\[
\widetilde{\mathcal{X}}\overset{d_{BL}}{\approx}\sigma\mathcal{N}(\mathbf{0},\mathbf{I}),
\]
where $\mathcal{N}(\mathbf{0},\mathbf{I})$ is the standard multi-variate Gaussian distribution with $p$ dimensions, $\sigma$ is the variance \chang{of the Gaussian distribution}, and $\overset{d_{BL}}{\approx}$ is the approximate
equality in distribution with respect to the conditional bounded-Lipschitz
distance.
\end{thm}
\begin{rem}
The conditional bounded-Lipschitz distance $d_{BL}$ is a distance metric, which can be used to measure the similarity between distributions. More detail on $d_{BL}$ can be found in \cite{RefWorks:531}. Here, the notion $\overset{d_{BL}}{\approx}$ is used to indicate that the distance between $\widetilde{\mathcal{X}}$ and $\sigma\mathcal{N}(\mathbf{0},\mathbf{I})$ is bounded by a small value \cite{RefWorks:439}.
\end{rem}

Theorem \ref{thm:proj_is_gauss} suggests that, under the regularity
conditions, most \emph{random orthonormal (RON) projections} of the
data are close to Gaussian in distribution. More specifically, Meckes
\cite{RefWorks:439} suggests that the Gaussian phenomenon generally
occurs for $p<\frac{2\log(m)}{\log(\log(m))}$. For example, if the
original dimension of the dataset is $m=100$, the projected data
would approach Gaussian with $p\leq13$. Intuitively, the regularity
conditions assure that the data are well-spread around the mean with
a finite second moment. The convergence to the standard Gaussian also
implies that the mean of $X$ is zero. However, this is less critical
since $d_{BL}$ is a distance metric \cite{RefWorks:531}, so mean-shift
can be shown to result in a Gaussian with a scaled mean-shift
(cf. \cite{RefWorks:531,RefWorks:513}).

\vspace{-0.5em}
\section{RON-Gauss: Exploiting the DFM Effect for Non-Interactive Differential
Privacy}

\label{sec:RON-Gauss}

Based on the DFM effect discussed in Section \ref{subsec:Diaconis-Freedman-Effect},
we present our approach for the non-interactive DP data release: the
\emph{RON-Gauss} model. In the subsequent discussion, we first give
an overview of the RON-Gauss model. Then, we discuss the approach used in RON-Gauss with corresponding  theoretical analyses. Finally, we present algorithms
to implement RON-Gauss for both unsupervised and supervised
learning tasks, and prove that the data generated from RON-Gauss preserve $\epsilon$-DP.

\vspace{-1em}
\subsection{Overview \label{subsec:Overview}}

RON-Gauss stands for \emph{\textbf{R}andom \textbf{O}rtho\textbf{N}ormal projection with \textbf{GAUSS}ian
generative model}. As its name suggests, RON-Gauss has two components
- dimensionality reduction (DR) via random orthonormal (RON) projection,
and parametric data modeling via the Gaussian generative model. Each
component plays an important role in RON-Gauss as follows.

The DR via random orthonormal (RON) projection has two purposes. First,
as previous works have shown \cite{RefWorks:313,RefWorks:314,RefWorks:337,RefWorks:338,RefWorks:339,RefWorks:224},
DR can reduce sensitivity of the data. This is true for many DR techniques.
However, in this work, we choose the RON projection due to the second purpose of DR
in the RON-Gauss model, i.e. to trigger the DFM effect. This is verified by Theorem \ref{thm:proj_is_gauss} as proved by Meckes \cite{RefWorks:439}.

The parametric modeling via the Gaussian generative model also has
two purposes. First, it allows us to fully exploit the DFM effect
since, unlike most practical data-analytic settings, we know the distribution
of the data from the effect. Second, it allows us to reduce the
difficult problem of non-interactive private data release into the
more amenable one of DP statistical estimation. Particularly, it reduces
the problem into that of privately estimating the mean and covariance
\textendash{} a problem which has seen success in DP literature (cf.
\cite{RefWorks:477,RefWorks:478,RefWorks:479,RefWorks:185,RefWorks:194,RefWorks:195,RefWorks:249,RefWorks:220}).

Combining these two components is crucial for getting high utility
from the released data, as we will demonstrate in our experiments, and we highlight the main differences between our approach and
previous works in non-interactive DP data release as follows.
\begin{itemize}
\item Although prior works have used DR for improving
utility of the released data (cf. Section \ref{subsed:dim_for_improve_utility}),
these works do not use the Gaussian generative model. Hence, they do not
fully exploit the DFM effect. 
\item Similarly, there have been prior works that use generative models for synthesizing
private data (cf. Section \ref{subsec:generative_models}). However,
without DR, the sensitivity is generally large for high-dimensional
data, and, more importantly, the DFM effect does not apply.
\item Unlike the work by Blocki et al.
\cite{RefWorks:313}, we do not use random projection to provide DP.
In RON-Gauss, DP is provided after the projection via the Laplace
mechanism \emph{on the Gaussian generative model}. 
\item Unlike previous
works that use i.i.d. Gaussian or Bernoulli random projection \cite{RefWorks:314,RefWorks:338,RefWorks:400,RefWorks:337},
we use RON projection, which has been proved to be suitable
for the DFM effect (Theorem \ref{thm:proj_is_gauss}).
\end{itemize}

\chang{Finally, we acknowledge that although RON-Gauss is designed to be application-independent, it may not be suitable for every task. In this work, we focus on popular machine learning tasks including clustering, regression, and classification. As discussed in Section \ref{subsec:gaussian_model}, many machine learning algorithms implicitly or explicitly utilize the Gaussian model, so RON-Gauss is generally suitable for these applications.
}

\vspace{-1em}
\subsection{Approach and Theoretical Analysis\label{subsec:Theoretical-Basis}}

The RON-Gauss model uses the following steps:
\begin{enumerate}
\item Pre-processing to satisfy the conditions for the DFM effect (Theorem
\ref{thm:proj_is_gauss}).
\begin{enumerate}
\item Pre-normalization.
\item Data centering.
\item Data re-normalization.
\end{enumerate}
\item RON projection.
\item Gaussian generative model estimation.
\end{enumerate}
\vspace{-1em}
We provide the detail of each process as follows.

\vspace{-1em}
\subsubsection{Data Pre-Processing}

Given a dataset with $n$ samples
and $m$ features, to utilize the DFM effect, we want to ensure that
the data satisfy the regularity conditions of Theorem \ref{thm:proj_is_gauss}.
We show that the following \emph{sample-wise normalization} ensures
the conditions are satisfied.
\begin{lem}[Sample-wise normalization]
\label{lem:normalization}Let $D\in\mathbb{R}^{m}$ be data drawn
i.i.d. from a distribution $\mathcal{D}$. Let $X\in\mathbb{R}^{m}$
be derived from $D$ by the sample-wise normalization\footnote{Here, it is implicitly assumed that $\left\Vert D\right\Vert$ is finite, which is typically the case when we are given a training dataset.}:
\[
X=\frac{D}{\left\Vert D\right\Vert }.
\]
Then, 
\begin{align*}
\mathbb{E}\left[\left\Vert X\right\Vert ^{2}\right] & =1,\\
\sup_{\mathbf{v}\in\mathbb{S}^{m-1}}\mathbb{E}\left\langle \mathbf{v},X\right\rangle ^{2} & \leq1,\\
\mathbb{E}\left[\left|\left\Vert X\right\Vert ^{2}\sigma^{-2}-m\right|\right] & =\left|\sigma^{-2}-m\right|.
\end{align*}
\end{lem}
\begin{proof}
The first and third equalities are obvious by observing that $\left\Vert X\right\Vert =1$.
The second inequality follows from the Cauchy\textendash Schwarz
inequality \cite{RefWorks:525,RefWorks:526} as follows. $\mathbb{E}\left\langle \mathbf{v},X\right\rangle ^{2}\leq\mathbb{E}\left[\left\Vert \mathbf{v}\right\Vert ^{2}\left\Vert X\right\Vert ^{2}\right]=1$,
since $\mathbf{v}$ is on the surface of the unit sphere. 
\end{proof}
From this lemma, we can verify that the sample-wise normalization
satisfies the regularity conditions in Theorem \ref{thm:proj_is_gauss}
simply by considering $\sigma=1/\sqrt{m}$. Recall from Theorem \ref{thm:proj_is_gauss} that the choice of $\sigma$ indicates what the variance of the projected data will be. In other words, a low value of $\sigma$ geometrically signifies a narrow bell curve of Gaussian. Hence, in our application,
this is the normalization we employ. However, we observe that this
normalization has an effect of placing all data samples onto the surface
of the sphere, i.e. $\mathbf{x}_{i}\in\mathbb{S}^{m-1}$. Hence, it
is beneficial to center the data before performing this normalization
to ensure that the data remain well-spread after the normalization.
For this reason, data pre-processing for RON-Gauss consists of three
steps \textendash{} pre-normalization, data centering, and data re-normalization.
As we will discuss shortly, the pre-normalization is to aid with the
sensitivity derivation of the sample mean used in the centering process, while the re-normalization
is to ensure the regularity conditions in Theorem \ref{thm:proj_is_gauss}
is satisfied before the projection. These three steps are discussed in detail as follows.

\vspace{-1em}
\paragraph{Pre-Normalization}

We start with a given dataset $\mathbf{X}\in\mathbb{R}^{m\times n}$
with $n$ samples and $m$ features, and perform the preliminary sample-wise
normalization as follows.
\[
\mathbf{x}_{i}:=\frac{\mathbf{x}_{i}}{\left\Vert \mathbf{x}_{i}\right\Vert },
\]
for all $\mathbf{x}_{i}\in\mathbf{X}$. This normalization ensures
that $\left\Vert \mathbf{x}_{i}\right\Vert =1$ for every sample,
which will be important for the derivation of the $L_{1}$-sensitivity in the next step.

\vspace{-1em}
\paragraph{Data Centering}

Data centering is performed before RON projection in order to reduce
the bias of the covariance estimation for the Gaussian generative
model and to ensure that the data are well-spread.
Data centering is achieved simply by subtracting the DP-mean of the
dataset. Given the pre-normalized dataset, $\{\mathbf{x}_{i}\in\mathbb{R}^{m};\left\Vert \mathbf{x}_{i}\right\Vert =1\}_{i=1}^{n}$,
the sample mean is $\boldsymbol{\mu}=\frac{1}{n}\sum_{i=1}^{n}\mathbf{x}_{i}$,
and the $L_{1}$-sensitivity of the sample mean can be computed as
follows.
\begin{lem}
\label{lem:sens_mean}Given a sample-wise normalized dataset $\mathbf{X}\in\mathbb{R}^{m\times n}$,
the $L_{1}$-sensitivity of the sample mean is $s(f)=2\sqrt{m}/n$.
\end{lem}
\begin{proof}
For neighboring datasets $\mathbf{X}$, $\mathbf{X}'$,
\begin{align*}
s(f) & =\sup_{\mathbf{X},\mathbf{X}'}\left\Vert f(\mathbf{X})-f(\mathbf{X}')\right\Vert _{1}\\
 & =\sup\frac{1}{n}\left\Vert \mathbf{x}_{i}-\mathbf{x}_{i}'\right\Vert _{1}\leq\sup\frac{\sqrt{m}}{n}\left\Vert \mathbf{x}_{i}-\mathbf{x}_{i}'\right\Vert _{F}\\
 & \leq\sup\frac{\sqrt{m}}{n}(\left\Vert \mathbf{x}_{i}\right\Vert _{F}+\left\Vert \mathbf{x}_{i}'\right\Vert _{F})=\frac{2\sqrt{m}}{n},
\end{align*}
\chang{where the first inequality uses the norm relation \cite[page 333]{RefWorks:208}.}
\end{proof}
With the $L_{1}$-sensitivity of the sample mean, we can then derive
DP-mean $\boldsymbol{\mu}^{DP}$ via the Laplace mechanism (Theorem
\ref{thm:laplace_matrix}), and perform data centering by $\bar{\mathbf{X}}=\mathbf{X}-\boldsymbol{\mu}^{DP}\mathbf{1}^{T}$,
where $\mathbf{1}$ is the vector with all ones. We note that, although
the mean is DP-protected, the centered data are not DP-protected,
so they cannot be released. RON-Gauss only uses the centered data
to estimate the covariance, which is then DP-protected. Hence, the centered data are never published.
In addition, as will be important to the DP analysis later, we
note that this centering process ensures that any neighboring datasets
would be centered by the same mean. Hence, neighboring $\bar{\mathbf{X}}$
and $\bar{\mathbf{X}}'$ would still differ by only one record.

\vspace{-1em}
\paragraph{Data Re-Normalization}

After adjusting the mean, the centered dataset $\bar{\mathbf{X}}$
would likely not remain normalized. Hence, to ensure the regularity
conditions in Theorem \ref{thm:proj_is_gauss}, we re-normalize the
data after the centering process using the same sample-wise normalization
scheme, i.e.
\[
\bar{\mathbf{x}}_{i}:=\frac{\bar{\mathbf{x}}_{i}}{\left\Vert \bar{\mathbf{x}}_{i}\right\Vert }.
\]
Hence, we again have $\left\Vert \bar{\mathbf{x}}_{i}\right\Vert =1$
for every sample. In addition, neighboring datasets still differ by
only one sample since the normalization factor only depends on the
corresponding sample, but not on any other sample.

\vspace{-1em}
\paragraph{Summary}
\noun{Data\_Preprocessing} (Algorithm \ref{alg:Data-Preprocessing})
summarizes these steps for pre-processing the data, which include pre-normalizing,
centering, and re-normalizing. The DP mean derivation uses the Laplace
mechanism with the sensitivity in Lemma \ref{lem:sens_mean}. \chang{If needed, the DP mean can also be acquired from this algorithm.}

\begin{algorithm}
{\small \par
\textbf{Input}{: dataset $\mathbf{X}\in\mathbb{R}^{m\times n}$
and $\epsilon_{\mu}>0$.}{\par}
\vspace{0.5em}
\begin{enumerate}
\item {Pre-normalize: $\mathbf{x}_{i}:=\mathbf{x}_{i}/\left\Vert \mathbf{x}_{i}\right\Vert $
for all $\mathbf{x}_{i}\in\mathbf{X}$.}{\par}
\item {Derive the DP mean: $\boldsymbol{\mu}^{DP}=(\frac{1}{n}\sum_{i=1}^{n}\mathbf{x}_{i})+Z$,
where $z_{j}(i)$ is drawn i.i.d. from $Lap(2\sqrt{m}/n\epsilon_{\mu})$.}{\par}
\item {Center the data: $\bar{\mathbf{X}}=\mathbf{X}-\boldsymbol{\mu}^{DP}\mathbf{1}^{T}$.}{\par}
\item {Re-normalize: $\bar{\mathbf{x}}_{i}=\bar{\mathbf{x}}_{i}/\left\Vert \bar{\mathbf{x}}_{i}\right\Vert $
for all $\bar{\mathbf{x}}_{i}\in\bar{\mathbf{X}}$.
}{\par}
\end{enumerate}
\vspace{-0.5em}
\textbf{Output}{: $\bar{\mathbf{X}}$, $\boldsymbol{\mu}^{DP}$. }{\par}
}
\caption{\noun{Data\_Preprocessing} \label{alg:Data-Preprocessing}}

\end{algorithm}

\vspace{-1em}
\subsubsection{RON Projection}

As shown in Lemma \ref{lem:normalization}, after the pre-processing
steps, $\bar{\mathbf{X}}$ can be shown to be in a form that complies
with the regularity conditions in Theorem \ref{thm:proj_is_gauss}.
The next step is to project $\bar{\mathbf{X}}$ onto a low-dimensional
space using the random orthonormal (RON) projection matrix $\mathbf{W}\in\mathbb{R}^{m\times p}:\mathbf{W}^{T}\mathbf{W}=\mathbf{I}$.
The RON projection matrix is derived independently of the dataset,
so it does not leak privacy, and no privacy budget is needed for its
acquisition.

The projection is done via the linear transformation $\widetilde{\mathbf{x}_{i}}=\mathbf{W}^{T}\bar{\mathbf{x}}_{i}\in\mathbb{R}^{p}$
for each sample, or, equivalently, in the matrix notation $\widetilde{\mathbf{X}}=\mathbf{W}^{T}\mathbf{X}\in\mathbb{R}^{p\times n}$.
Since the projection is done sample-wise, the neighboring datasets
would still differ by only one sample after the projection. This property
will be important to the DP analysis of the RON projection later.
In addition, the other important theoretical aspect of the RON projection
step is the bound on the projected data. This is provided by the following
lemma.
\begin{lem} \label{lem:proj_range} 
Given a normalized data sample $\bar{\mathbf{x}}\in\mathbb{R}^{m}$
and a random orthonormal (RON) projection matrix, $\mathbf{W}\in\mathbb{R}^{m\times p}:\mathbf{W}^{T}\mathbf{W}=\mathbf{I}$,
let $\widetilde{\mathbf{x}}=\mathbf{W}^{T}\bar{\mathbf{x}}$ be the projection
via $\mathbf{W}$. Then, $\left\Vert \widetilde{\mathbf{x}}\right\Vert _{F}\leq1$. 
\end{lem}
\begin{proof}
The proof is provided in Appendix \ref{sec:proof_lem_proj_range}.
\end{proof}

Lemma \ref{lem:proj_range} indicates that the RON projection of the
normalized data does not change their Frobenius norm. This will be critical in the DP analysis of RON-Gauss algorithms in the next step.

\noun{RON\_Projection} (Algorithm \ref{alg:RON_Projection}) summarizes
the current step that projects the pre-processed data onto a lower dimension
$p$ via RON projection. The RON projection matrix can be derived
efficiently via the QR factorization \cite{RefWorks:533,RefWorks:534},
as shown in the algorithm. Specifically, the RON projection matrix
is constructed by stacking side-by-side $p$ column vectors of the
unitary matrix $\mathbf{Q}$ from the QR factorization. Then, the projection is done via the linear operation
{\small{}$\widetilde{\mathbf{X}}=\mathbf{W}^{T}\bar{\mathbf{X}}\in\mathbb{R}^{p\times n}$.}{\small} \chang{If needed, the RON projection matrix can also be acquired from the output of the algorithm.}

We further note that the projected data $\widetilde{\mathbf{X}}$
are still not DP-protected. This is one key difference between our
work and that of Blocki et al. \cite{RefWorks:313}, which uses random
projection to directly provide DP. In our work, $\widetilde{\mathbf{X}}$
is never released, and we only use it to estimate the covariance of
the Gaussian generative model in the next step. The DP protection
in our work is provided in this next step on the Gaussian generative model.

\begin{algorithm}
{\small \par
\textbf{Input}{: pre-processed dataset $\bar{\mathbf{X}}\in\mathbb{R}^{m\times n}$,
and dimension $p<m$.
}{\par}
\vspace{0.5em}
\begin{enumerate}
\item {Form a matrix $\mathbf{A}\in\mathbb{R}^{m\times m}$ whose elements are drawn i.i.d. from the uniform distribution.}{\par}
\item {Factorize $\mathbf{A}$ via the QR factorization as $\mathbf{A}=\mathbf{Q}\mathbf{R}$,
where $\mathbf{Q}\in\mathbb{R}^{m\times m}:\mathbf{Q}^{T}\mathbf{Q}=\mathbf{I}$.}{\par}
\item {Construct a RON projection matrix $\mathbf{W}=[\mathbf{q}_{1},\ldots,\mathbf{q}_{p}]\in\mathbb{R}^{m\times p}$.}{ \par}
\item {Project the data: $\widetilde{\mathbf{X}}=\mathbf{W}^{T}\bar{\mathbf{X}}\in\mathbb{R}^{p\times n}$.\vspace{-0.5em}
}{\par}
\end{enumerate}
\textbf{Output}{: $\widetilde{\mathbf{X}}\in\mathbb{R}^{p\times n}$,
$\mathbf{W}$.}{ \par}
}
\caption{\noun{RON\_Projection \label{alg:RON_Projection}}}

\end{algorithm}

\vspace{-1em}
\subsubsection{Gaussian Generative Model Estimation}

This step constructs the Gaussian generative model, which is where
the DP protection in RON-Gauss is provided. \chang{We emphasize that RON-Gauss is an output-perturbation algorithm, and we employ the standard DP threat model, i.e. the RON-Gauss algorithm is run by a trusted entity and only the output of the algorithm is available to the public.} The DP-protected Gaussian
generative model is then used to synthesize DP dataset for the non-interactive
DP data release setting we consider. \chang{Synthesizing DP data from a parametric model, as opposed to releasing the model itself, has two benefits. First, existing machine learning software can readily be used with the DP data as if they were the real data. Second, it presents an additional challenge for an attacker aiming to perform inference attacks, since the attacker would also need to estimate the model parameters from the released data, incurring further errors.}

Before delving into the detail
of this step, there is an important distinction to be made about the
data-analytic problems we consider. Since machine learning is currently
the prominent tool in data analysis, we follow the convention in machine
learning and consider two classes of problems \textendash{} \emph{unsupervised
learning} and \emph{supervised learning}. The main difference between
the two is that, in the latter, in addition to the feature data in
$\widetilde{\mathbf{X}}$, the \emph{teacher value} or \emph{training
label} $\mathbf{y}\in\mathbb{R}^{n}$ is also required to guide the
data-analytic process. Hence, in the subsequent analysis, we first consider the simpler class
of unsupervised learning, and then, show a simple modification to
include the teacher value for the supervised learning. Additionally,
we conclude with an extension of RON-Gauss to the Gaussian
Mixture Model.

 \vspace{-1em}
\paragraph{Unsupervised Learning}

The unsupervised learning problems do not require the training label, so the Gaussian generative model only
needs to synthesize DP-protected $\widetilde{\mathbf{X}}$. The main
parameter for the Gaussian generative model is the covariance matrix
$\boldsymbol{\Sigma}$, so we need to estimate $\boldsymbol{\Sigma}$
from $\widetilde{\mathbf{X}}$. We use the following formulation for
the sample covariance:
\begin{equation}
\boldsymbol{\Sigma}=\frac{1}{n}\widetilde{\mathbf{X}}\widetilde{\mathbf{X}}^{T}=\frac{1}{n}\sum_{i=1}^{n}\widetilde{\mathbf{x}}_{i}\widetilde{\mathbf{x}}_{i}^{T}\in\mathbb{R}^{p\times p}.\label{eq:cov_est}
\end{equation}
We note that this estimate may be statistically biased since the mean
may not necessarily be zero after the re-normalization. However, this
formulation would yield significantly lower sensitivity than its unbiased
counterpart. This is due to the observation that only one summand can change
for neighboring datasets since, as mentioned in the previous step,
neighboring projected datasets $\widetilde{\mathbf{X}}, \widetilde{\mathbf{X}'}$ still differ by only one sample.
Hence, we are willing to trade the bias for a much lower sensitivity.
In Appendix \ref{sec:sens_mle_cov}, we specifically show that the
saving in sensitivity by our formulation is in the order of $n$,
compared to the MLE of the covariance. Clearly, for large datasets,
this is significant and can be the difference between usable and unusable
models.

Next, we derive the sensitivity of the covariance estimate in Eq.
\eqref{eq:cov_est} as follows.
\begin{lem}
\label{lem:sens_cov} Given a dataset $\mathbf{X}\in\mathbb{R}^{m\times n}$,
let $\widetilde{\mathbf{X}}$ be the pre-processed and RON-projected
dataset via \noun{Data\_Preprocessing} and \noun{RON\_Projection}.
Then, the covariance $\boldsymbol{\Sigma}\in\mathbb{R}^{p\times p}$
in Eq. \eqref{eq:cov_est} has the $L_{1}$-sensitivity of $2\sqrt{p}/n$.
\end{lem}
\begin{proof}
For neighboring datasets $\mathbf{X},\mathbf{X}'$,
\begin{align*}
s(f)= & \sup\left\Vert \frac{1}{n}\widetilde{\mathbf{X}}\widetilde{\mathbf{X}}^{T}-\frac{1}{n}\widetilde{\mathbf{X}'}\widetilde{\mathbf{X}'}^{T}\right\Vert _{1}\\
= & \sup\left\Vert \frac{1}{n}\mathbf{W}^{T}[\mathbf{X}-\boldsymbol{\mu}^{DP}\mathbf{1}^{T}][\mathbf{X}-\boldsymbol{\mu}^{DP}\mathbf{1}^{T}]^{T}\mathbf{W}\right.\\
 & -\left.\frac{1}{n}\mathbf{W}^{T}[\mathbf{X}'-\boldsymbol{\mu}^{DP}\mathbf{1}^{T}][\mathbf{X}'-\boldsymbol{\mu}^{DP}\mathbf{1}^{T}]^{T}\mathbf{W}\right\Vert _{1},\\
= & \sup\frac{1}{n}\left\Vert \sum_{i=1}^{n}\mathbf{W}^{T}[\mathbf{x}_{i}-\boldsymbol{\mu}^{DP}][\mathbf{x}_{i}-\boldsymbol{\mu}^{DP}]^{T}\mathbf{W}\right.\\
 & -\left.\sum_{i=1}^{n}\mathbf{W}^{T}[\mathbf{x}_{i}'-\boldsymbol{\mu}^{DP}][\mathbf{x}_{i}'-\boldsymbol{\mu}^{DP}]^{T}\mathbf{W}\right\Vert _{1},
\end{align*}
where the second equality is simply from the definition of $\widetilde{\mathbf{X}}$
through\noun{ Data\_Preprocessing} and \noun{RON\_Projection}. Since
all of the summands are the same except for one in the neighboring
datasets, we have
\begin{align*}
s(f)= & \sup\frac{1}{n}\left\Vert \mathbf{W}^{T}[\mathbf{x}_{i}-\boldsymbol{\mu}^{DP}][\mathbf{x}_{i}-\boldsymbol{\mu}^{DP}]^{T}\mathbf{W}\right.\\
 & -\left.\mathbf{W}^{T}[\mathbf{x}_{i}'-\boldsymbol{\mu}^{DP}][\mathbf{x}_{i}'-\boldsymbol{\mu}^{DP}]^{T}\mathbf{W}\right\Vert _{1}.
\end{align*}
Then, to simplify the notation and to apply Lemma \ref{lem:proj_range},
we note that $\widetilde{\mathbf{x}}_{i}=\mathbf{W}^{T}[\mathbf{x}_{i}-\boldsymbol{\mu}^{DP}]$
by definition. Hence, 
\begin{align*}
s(f) & =\sup\frac{1}{n}\left\Vert \widetilde{\mathbf{x}}_{i}\widetilde{\mathbf{x}}_{i}^{T}-\widetilde{\mathbf{x}'}_{i}\widetilde{\mathbf{x}'}_{i}^{T}\right\Vert _{1}\\
 & \leq\sup\frac{\sqrt{p}}{n}\left\Vert \widetilde{\mathbf{x}}_{i}\widetilde{\mathbf{x}}_{i}^{T}-\widetilde{\mathbf{x}'}_{i}\widetilde{\mathbf{x}'}_{i}^{T}\right\Vert _{F}\\
 & \leq\sup\frac{2\sqrt{p}}{n}\left\Vert \widetilde{\mathbf{x}}_{i}\widetilde{\mathbf{x}}_{i}^{T}\right\Vert _{F} \leq\sup\frac{2\sqrt{p}}{n}\left\Vert \widetilde{\mathbf{x}}_{i}\right\Vert _{F}^{2}=\frac{2\sqrt{p}}{n},
\end{align*}
where the first inequality uses the norm relation \cite[page 333]{RefWorks:208},
and the last equality uses Lemma \ref{lem:proj_range}.
\end{proof}

Lemma \ref{lem:sens_cov} provides an important insight into the RON-Gauss
model. As discussed in the overview (Section \ref{subsec:Overview}),
the RON projection step of RON-Gauss serves two purposes \textendash{}
to initiate the DFM effect, and to reduce the sensitivity of the model.
The latter purpose can clearly be observed from Lemma \ref{lem:sens_cov}.
Specifically, Lemma \ref{lem:sens_cov} indicates that the $L_{1}$-sensitivity
of the main parameter of the RON-Gauss model, i.e. the covariance,
reduces as the dimension $p$ reduces. This is particularly attractive
when the original data are very high-dimensional as the noise added
by the Laplace mechanism could be greatly reduced. For example, for
the original data with 100 dimensions, the RON projection onto a 10-dimensional
subspace would reduce the sensitivity by about 3x.

With the $L_{1}$-sensitivity derived, the Laplace mechanism in Theorem
\ref{thm:laplace_matrix} can be used to derive the DP covariance
matrix: $\boldsymbol{\Sigma}^{DP}$. With $\boldsymbol{\Sigma}^{DP}$, RON-Gauss then generates
the synthetic DP data from $\mathcal{N}(\mathbf{0},\boldsymbol{\Sigma}^{DP})$.
If the mean is needed, we can readily use $\mathbf{W}^{T}\boldsymbol{\mu}^{DP}$,
which already satisfies DP due to the post-processing invariance of
DP \cite{RefWorks:185}. This completes the RON-Gauss model for unsupervised
learning. 

\begin{algorithm}
{\small \par
\textbf{Input}{: dataset $\mathbf{X}\in\mathbb{R}^{m\times n}$,
dimension $p<m$, and $\epsilon_{\mu},\epsilon_{\Sigma}>0$. \vspace{0.5em}
}{\par}
\begin{enumerate}
\item {Obtain the pre-processed data $\bar{\mathbf{X}}\in\mathbb{R}^{m\times n}$
from }\noun{Data\_Preprocessing}{ with inputs $\mathbf{X}$
and $\epsilon_{\mu}$.}{ \par}
\item {Obtain the RON-projected data $\widetilde{\mathbf{X}}\in\mathbb{R}^{p\times n}$
from }\noun{RON\_Projection}{ with inputs $\bar{\mathbf{X}}$
and $p$.}{ \par}
\item {Derive the DP covariance: $\boldsymbol{\Sigma}^{DP}=(\frac{1}{n}\widetilde{\mathbf{X}}\widetilde{\mathbf{X}}^{T})+Z$,
where $z_{j}(i)$ is drawn i.i.d. from $Lap(2\sqrt{p}/n\epsilon_{\Sigma})$.}{\par}
\item {Synthesize DP data $\mathbf{x}_{i}^{DP}\in\mathbb{R}^{p}$
by drawing samples from $\mathcal{N}(\mathbf{0},\boldsymbol{\Sigma}^{DP})$.
\vspace{-0.5em}
}{ \par}
\end{enumerate}
\textbf{Output}{: $\{\mathbf{x}_{1}^{DP},\ldots,\mathbf{x}_{n'}^{DP}\}$.
}
}
\caption{RON-Gauss for unsupervised learning \label{alg:unsupervised}}
\end{algorithm}

Algorithm \ref{alg:unsupervised} summarizes the RON-Gauss model for
unsupervised learning. First, the data are pre-processed via \noun{Data\_Preprocessing}
(Algorithm \ref{alg:Data-Preprocessing}). The DP mean derivation
in \noun{Data-Preprocessing} spends the privacy budget $\epsilon_{\mu}$.
Second, the data are projected onto a lower dimension $p$ via \noun{RON\_Projection}.
Third, the algorithm derives the DP covariance using the Laplace mechanism
with the sensitivity derived in Lemma \ref{lem:sens_cov}. Finally,
the algorithm synthesizes DP data by drawing samples from the Gaussian
generative model parametrized by the DP covariance. 
We conclude the discussion on RON-Gauss for unsupervised learning with the DP analysis of Algorithm \ref{alg:unsupervised}.
\begin{thm}
\label{thm:unsupervised_priv}Algorithm \ref{alg:unsupervised} preserves
$(\epsilon_{\mu}+\epsilon_{\Sigma})$-differential privacy. 
\end{thm}
\begin{proof}
The proof follows the following induction. The DP data $\mathbf{x}_{i}^{DP}$
are derived from only one source, i.e. the Gaussian generative model
of RON-Gauss. Based on the post-processing invariance of DP, if the
model is DP-protected, then the released data are also similarly DP-protected.
The Gaussian generative model is parametrized by $\boldsymbol{\Sigma}^{DP}$,
which is DP-protected. Specifically, the $\boldsymbol{\Sigma}^{DP}$
computation in step 3(a) spends $\epsilon_{\Sigma}$ privacy budget
with the Laplace mechanism, according to Theorem \ref{thm:laplace_matrix}.
However, the centering process in step 1(c) also spends $\epsilon_{\mu}$
privacy budget on the Laplace mechanism to derive $\boldsymbol{\mu}^{DP}$,
which assists in the $\boldsymbol{\Sigma}^{DP}$ derivation process.
Due to the serial composition theorem \cite{RefWorks:185}, the two
privacy budgets add up. Hence, $\boldsymbol{\Sigma}^{DP}$ preserves
$(\epsilon_{\mu}+\epsilon_{\Sigma})$-differential privacy, and, consequently,
the Gaussian generative model preserves $(\epsilon_{\mu}+\epsilon_{\Sigma})$-differential
privacy, so do the synthesized data.
\end{proof}

\vspace{-0.5em}
\paragraph{Supervised Learning}

The unsupervised learning does not involve the guidance from the training label. However, in supervised learning, the training label $\mathbf{y}$ is also required. Hence, the
Gaussian generative model needs to be modified to incorporate the
training label into the model in order to synthesize both DP-protected
$\widetilde{\mathbf{X}}$ and $\mathbf{y}$. 

A simple method to incorporate the training label into the Gaussian
generative model is to treat it as another feature. However, when
RON projection is applied, it should only be applied to the feature
data, but \emph{not to the training label}. This is because when the
projection is applied, each induced feature is a linear combination
of all original features. Therefore, if RON projection is also applied
to the training label, the integrity of the training label would be
spoiled. In other words, we may not be able to extract the training
label from the projected data. Thus, to preserve the integrity of
the training label, it should \emph{not} be modified by the RON projection
process.

In RON-Gauss, the aforementioned challenge in supervised learning
is navigated by augmenting the data matrix with the training label as $\mathbf{X}_{a}=\begin{bmatrix}\widetilde{\mathbf{X}}\\
\mathbf{y}^{T}
\end{bmatrix}\in\mathbb{R}^{(p+1)\times n}$. Then, the augmented covariance matrix can be written in block form
as,
\begin{equation}
\boldsymbol{\Sigma}_{a}=\frac{1}{n}\begin{bmatrix}\widetilde{\mathbf{X}}\\
\mathbf{y}^{T}
\end{bmatrix}\begin{bmatrix}\widetilde{\mathbf{X}}^{T} & \mathbf{y}\end{bmatrix}=\frac{1}{n}\begin{bmatrix}\widetilde{\mathbf{X}}\widetilde{\mathbf{X}}^{T} & \widetilde{\mathbf{X}}\mathbf{y}\\
\mathbf{y}^{T}\widetilde{\mathbf{X}}^{T} & \mathbf{y}^{T}\mathbf{y}
\end{bmatrix}.\label{eq:augmented_cov}
\end{equation}
This can then be used in a similar fashion to the covariance
matrix in Eq. \eqref{eq:cov_est} for unsupervised learning. We note
that, again, this may not be an unbiased estimate of the augmented
covariance matrix since the mean may not necessarily be zero, but,
similar to the unsupervised learning design, it has significantly
lower sensitivity than the unbiased counterpart. Therefore, we are willing
to trade the bias for achieving small sensitivity. Given the training
label with bounded value\footnote{As suggested by Liu \cite{RefWorks:372}, real-world data are often
bounded, and the bounded-valued
assumption is often made in DP analysis for multi-dimensional query
(cf. \cite{RefWorks:195,RefWorks:178,RefWorks:338,RefWorks:249}).} $\mathbf{y}\in[-a,a]^{n}$, the sensitivity of the augmented covariance matrix can be derived as follows.
\begin{lem}
\label{lem:sens_aug_cov}The $L_{1}$-sensitivity of the augmented
covariance matrix in Eq. \eqref{eq:augmented_cov} is $\frac{2\sqrt{p}+4a\sqrt{p}+a^{2}}{n}$.
\end{lem}
\begin{proof}
For neighboring datasets $\mathbf{X},\mathbf{X}'$, the sensitivity
is
\begin{align*}
S(\boldsymbol{\Sigma}_{a})= & \sup\frac{1}{n}\left\Vert \begin{bmatrix}\widetilde{\mathbf{X}}\widetilde{\mathbf{X}}^{T} & \widetilde{\mathbf{X}}\mathbf{y}\\
\mathbf{y}^{T}\widetilde{\mathbf{X}}^{T} & \mathbf{y}^{T}\mathbf{y}
\end{bmatrix}-\begin{bmatrix}\widetilde{\mathbf{X}'}\widetilde{\mathbf{X}'}^{T} & \widetilde{\mathbf{X}'}\mathbf{y}'\\
\mathbf{y}'^{T}\widetilde{\mathbf{X}'}^{T} & \mathbf{y}'^{T}\mathbf{y}'
\end{bmatrix}\right\Vert _{1},\\
= & \sup\frac{1}{n}(\left\Vert \widetilde{\mathbf{X}}\widetilde{\mathbf{X}}^{T}-\widetilde{\mathbf{X}'}\widetilde{\mathbf{X}'}^{T}\right\Vert _{1}+2\left\Vert \widetilde{\mathbf{X}}\mathbf{y}-\widetilde{\mathbf{X}'}\mathbf{y}'\right\Vert _{1}\\
 & +\left\Vert \mathbf{y}^{T}\mathbf{y}-\mathbf{y}'^{T}\mathbf{y}'\right\Vert _{1}).
\end{align*}
The proof then considers each summand separately. The first summand
is the sensitivity of $\boldsymbol{\Sigma}$ in Eq. \eqref{eq:cov_est},
so it is $2\sqrt{p}/n$. The last summand can be written as,
\begin{align*}
\sup\frac{\left\Vert \mathbf{y}^{T}\mathbf{y}-\mathbf{y}'^{T}\mathbf{y}'\right\Vert _{1}}{n} & =\sup\frac{\left\Vert \sum_{i=1}^{n}y(i)^{2}-\sum_{i=1}^{n}y'(i)^{2}\right\Vert _{1}}{n}\\
 & =\sup\frac{\left\Vert y(i)^{2}-y'(i)^{2}\right\Vert _{1}}{n}=\frac{a^{2}}{n},
\end{align*}
where the second equality is because only one element
in $\mathbf{y}$ and $\mathbf{y}'$ differs.

For the second summand, we have
\begin{align*}
\sup\frac{2\left\Vert \widetilde{\mathbf{X}}\mathbf{y}-\widetilde{\mathbf{X}'}\mathbf{y}'\right\Vert _{1}}{n}
 & =\sup\frac{2\left\Vert \sum_{i}\widetilde{\mathbf{x}_{i}}y(i)-\sum_{i}\widetilde{\mathbf{x}'_{i}}y'(i)\right\Vert _{1}}{n}\\
 & =\sup\frac{2\left\Vert \widetilde{\mathbf{x}_{i}}y(i)-\widetilde{\mathbf{x}'_{i}}y'(i)\right\Vert _{1}}{n}\\
 & \leq\sup\frac{2(\left\Vert \widetilde{\mathbf{x}_{i}}y(i)\right\Vert _{1}+\left\Vert \widetilde{\mathbf{x}'_{i}}y'(i)\right\Vert _{1})}{n}\\
 & \leq\sup\frac{2\sqrt{p}(\left\Vert \widetilde{\mathbf{x}_{i}}y(i)\right\Vert _{F}+\left\Vert \widetilde{\mathbf{x}'_{i}}y'(i)\right\Vert _{F})}{n}\\
 & \leq\frac{2\sqrt{p}(2a)}{n}=\frac{4a\sqrt{p}}{n},
\end{align*}
where the second line is from the fact that the other $n-1$ terms
are similar for neighboring datasets. By combing the three summands,
we have completed the proof. 
\end{proof}
As the $L_{1}$-sensitivity of $\boldsymbol{\Sigma}_{a}$ has been
derived, we can use the Laplace mechanism in Theorem \ref{thm:laplace_matrix}
to acquire the DP augmented covariance matrix: $\boldsymbol{\Sigma}_{a}^{DP}$.
Then, similar to the unsupervised learning, RON-Gauss generates the
synthetic DP data \textendash{} which include both the feature data
and the training label \textendash{} from $\mathcal{N}(\mathbf{0},\boldsymbol{\Sigma}_{a}^{DP})$.
Notice that the only difference between the RON-Gauss model for unsupervised
learning and supervised learning is the use of $\boldsymbol{\Sigma}^{DP}$
(Eq. \eqref{eq:cov_est}) and $\boldsymbol{\Sigma}_{a}^{DP}$ (Eq.
\eqref{eq:augmented_cov}), respectively.

\begin{algorithm}[t]
{\small \par
\textbf{Input}{: dataset with training labels $\mathbf{X}\in\mathbb{R}^{m\times n},\mathbf{y}\in[-a,a]^{n}$,
dimension $p<m$, and $\epsilon_{\mu},\epsilon_{\Sigma}>0$. \vspace{0.5em}
}{ \par}
\begin{enumerate}
\item {Obtain the pre-processed data $\bar{\mathbf{X}}\in\mathbb{R}^{m\times n}$
from }\noun{Data\_Preprocessing}{ with inputs $\mathbf{X}$
and $\epsilon_{\mu}$.}{ \par}
\item {Obtain the RON-projected data $\widetilde{\mathbf{X}}\in\mathbb{R}^{p\times n}$
from }\noun{RON\_Projection}{ with inputs $\bar{\mathbf{X}}$
and $p$.}{ \par}
\item Form the augmented data matrix $\mathbf{X}_{a}=\begin{bmatrix}\widetilde{\mathbf{X}}\\
\mathbf{y}^{T}
\end{bmatrix}\in\mathbb{R}^{(p+1)\times n}$.
\item {Derive the DP augmented covariance: $\boldsymbol{\Sigma}_{a}^{DP}=(\frac{1}{n}\mathbf{X}_{a}\mathbf{X}_{a}^{T})+Z^{(\Sigma)}$,
where $z_{j}^{\Sigma}(i)$ is drawn i.i.d. from $Lap((2\sqrt{p}+4a\sqrt{p}+a^{2})/n\epsilon_{\Sigma})$.}{\par}
\item {Synthesize DP augmented data $\left[\begin{array}{c}
\mathbf{x}_{i}^{DP}\\
y(i)^{DP}
\end{array}\right]\in\mathbb{R}^{p+1}$ by drawing samples from $\mathcal{N}(\mathbf{0},\boldsymbol{\Sigma}_{a}^{DP})$.
\vspace{-0.5em}
}{\par}
\end{enumerate}
\textbf{Output}{: $\{\mathbf{x}_{1}^{DP},\ldots,\mathbf{x}_{n'}^{DP}\}$
with training label $\mathbf{y}^{DP}$. }
}
\caption{RON-Gauss for supervised learning \label{alg:RON-Gauss_supervised}}
\end{algorithm}

Algorithm \ref{alg:RON-Gauss_supervised} presents the RON-Gauss model
for supervised learning. The algorithm is similar to Algorithm
\ref{alg:unsupervised}. The only difference is the use of the augmented
covariance matrix in step 4 with the sensitivity from Lemma \ref{lem:sens_aug_cov} to incorporate the training label. As a result, Algorithm
\ref{alg:RON-Gauss_supervised} can synthesize both the DP feature
data $\mathbf{x}_{i}^{DP}$ and the DP training label $\mathbf{y}^{DP}$.
Finally, we present the privacy guarantee of Algorithm \ref{alg:RON-Gauss_supervised} as follows.

\begin{thm}
\label{thm:Algorithm2_priv}Algorithm \ref{alg:RON-Gauss_supervised}
preserves $(\epsilon_{\mu}+\epsilon_{\Sigma})$-differential privacy. 
\end{thm}
\begin{proof}
The proof mirrors that of Theorem \ref{thm:unsupervised_priv} but
uses the sensitivity of the augmented covariance in Lemma \ref{lem:sens_aug_cov}
instead.
\end{proof}

\vspace{-1em}
\paragraph{Extension to Gaussian Mixture Model}
\label{sec:gmm}

Algorithm \ref{alg:RON-Gauss_supervised} for supervised learning
uses the unimodal Gaussian generative model. \chang{The labels synthesized from this algorithm are numerical. In many applications, e.g. regression, this may already be effective. However, in some applications, e.g. \emph{classification}, it is desirable to synthesize the labels that are discrete or categorical.} To this end, we extend RON-Gauss to a multi-modal Gaussian generative model
using the \emph{Gaussian Mixture Model} (GMM) \cite{RefWorks:51,RefWorks:225}.
Conceptually, each mode of GMM can be used to capture the distribution
of the data in each class. Thus, the entire dataset is modeled by
the mixture of these modes. In fact, many classifiers such as Linear
Discriminant Analysis (LDA) \cite{RefWorks:292}, Bayes Net \cite{RefWorks:51},
and mixture of Gaussians \cite{RefWorks:51} also utilize this type
of generative model, so this GMM extension has historically been shown
to be effective for classification.

\begin{algorithm}
{\small \par
\textbf{Input}{: dataset $\mathbf{X}\in\mathbb{R}^{m\times n},\mathbf{y}\in\{c_{1},c_{2},\ldots,c_{L}\}^{n}$
, dimension $p<m$, and $\epsilon_{\boldsymbol{\mu}},\epsilon_{\boldsymbol{\Sigma}}>0$.\vspace{0.5em}
}{\par}
\begin{enumerate}
\item \textbf{for} {$c$ in $\{c_{1},\ldots,c_{L}\}$ }\textbf{\small{}do:}{\small{}\vspace{0.5em}
}{\par}
\begin{enumerate}
\item {Form $\mathbf{X}_{c}$, whose $n_{c}$ column vectors are
all samples in class $c$.}{\par}
\item {Obtain the pre-processed data $\bar{\mathbf{X}}_{c}\in\mathbb{R}^{m\times n_{c}}$
and the DP class-mean $\boldsymbol{\mu}_{c}^{DP}$ from }\noun{Data\_Preprocessing}{
with inputs $(\mathbf{X}_{c},\epsilon_{\mu})$.}{ \par}
\item {Obtain the RON-projected data $\widetilde{\mathbf{X}}_{c}\in\mathbb{R}^{p\times n_{c}}$
and the projection matrix $\mathbf{W}$ from }\noun{\small{}RON\_Projection}{
with inputs $\bar{\mathbf{X}}_{c}$ and $p$.}{ \par}
\item {Derive the DP class-covariance: $\boldsymbol{\Sigma}_{c}^{DP}=(\frac{1}{n_{c}}\widetilde{\mathbf{X}}_{c}\widetilde{\mathbf{X}}_{c}^{T})+Z$,
where $z_{j}(i)$ is drawn i.i.d. from $Lap(2\sqrt{p}/n_{c}\epsilon_{\Sigma})$.}{\par}
\item {Synthesize DP class-$c$ data $\mathbf{X}_{c}^{DP}$ by drawing
samples from $\mathcal{N}(\mathbf{W}^{T}\boldsymbol{\mu}_{c}^{DP},\boldsymbol{\Sigma}_{c}^{DP})$,
and assign $\mathbf{y}_{c}^{DP}=c$ for all samples in $\mathbf{X}_{c}^{DP}$.

}{ \par}
\end{enumerate}
\item {Let $\mathbf{X}^{DP}=[\mathbf{X}_{c_{1}}^{DP},\ldots,\mathbf{X}_{c_{L}}^{DP}]$.}{ \par}

\item {Let $\mathbf{y}^{DP}=[\mathbf{y}_{c_{1}}^{DP^{T}},\ldots,\mathbf{y}_{c_{L}}^{DP^{T}}]^{T}$.\vspace{-0.5em}
}{ \par}
\end{enumerate}
\textbf{Output}{: $\{\mathbf{X}^{DP},\mathbf{y}^{DP}\}$.
}
}
\caption{RON-Gauss' extension to GMM\label{alg:RON-Gauss-for-gmm}}
\end{algorithm}

In classification, the training label is categorical, i.e. $y\in\{c_{1},\ldots,c_{L}\}$
for $L$-class classification. Algorithm \ref{alg:RON-Gauss-for-gmm}
presents an extension of RON-Gauss to GMM. The algorithm iterates
through the data samples in each class. It derives DP samples for each class in
a similar procedure to Algorithm {\small{}\ref{alg:unsupervised} }with
one difference. For GMM, the data in each class are generated from
the Gaussian generative model with the mean equal to the RON-projected
DP class-mean, i.e. $\mathbf{W}^{T}\boldsymbol{\mu}_{c}^{DP}$. This
is to capture the multi-modal nature of GMM. Since every DP sample
drawn from each iteration of step 1 belongs to the same class, the
same training label is assigned for every synthesized sample. Finally, after iterating
through all classes, the algorithm stacks the DP samples and training
labels together before releasing the synthesized data.

We note that this algorithm assumes that each data sample belongs to
one class only, so each mode of Gaussian is derived from a disjoint
set of data. This is the common setting in supervised learning applications (cf. \cite{RefWorks:51,RefWorks:33,RefWorks:225,RefWorks:376}).
In addition, to comply with the bounded DP notion we adopt throughout
(cf. Remark \ref{rem:neighboring_notion}), the algorithm assumes
that the number of samples in each class $n_{c}$ is public information.
Finally, we present the DP analysis of Algorithm \ref{alg:RON-Gauss-for-gmm}
as follows.
\begin{thm}
Algorithm \ref{alg:RON-Gauss-for-gmm} preserves $(\epsilon_{\boldsymbol{\mu}}+\epsilon_{\boldsymbol{\Sigma}})$-differential
privacy. 
\end{thm}
\begin{proof}
Since the data partition is disjoint, and each class has the same domain, the privacy
budget used for each class does not add up from the parallel composition \cite{RefWorks:477,mcsherry2009privacy}. The proof then follows from
that of Theorem \ref{thm:unsupervised_priv}.
\end{proof}

\vspace{-1.5em}
\section{Experiments\label{sec:Experiments}}

We demonstrate that RON-Gauss is effective across a range of datasets
and machine learning tasks via three experiments. For the datasets,
we use a facial expression dataset~\cite{RefWorks:230}, a sensor
displacement dataset~\cite{RefWorks:207,RefWorks:358}, and a social
media dataset~\cite{RefWorks:240}. For the machine learning tasks,
we use the clustering, classification, and regression applications. In non-interactive DP, the aim is to release DP data such that the utility of the DP data closely resembles that of the original data. Hence, to evaluate the quality of the non-interactive DP algorithms based on the utility measure commonly used for the respective task (cf. Section \ref{subsec:Setups}). DP data with high quality, therefore, should provide the values of the utility measure close to that obtained from the non-private data, and our experiments show that RON-Gauss can achieve this objective. \chang{We note that we choose task-centric utility measures for our experiments since we want to evaluate the approach based on how much insight can be gained from the synthesized data with respect to each task. However, in other settings, task-independent evaluation metrics such as reconstruction error or mutual information could also be appropriate.} We also compare our work to four previous approaches that \emph{solely}
relied on either DR, or generative models. Table
\ref{tab:exp_setup_summary} summarizes the experimental setups, and
we discuss them in detail as follows.

\begin{table*}
\begin{centering}
\begin{tabular}{>{\centering}p{2.2cm}|>{\centering}p{2.5cm}|>{\centering}p{2cm}|>{\centering}p{2cm}|>{\centering}p{2.8cm}|c|>{\centering}p{1.5cm}}
\hline 
{Exp. } & {Dataset } & {Training Size } & {Feature Size } & {Metric } & {ML Alg. } & {DP Alg. }\tabularnewline
\hline 
\hline 
{Clustering } & {GFE \cite{RefWorks:230}} & {27,936 } & {301 } & {S.C. ($\uparrow$better) } & {K-Means } & {Alg. \ref{alg:unsupervised} }\tabularnewline
\hline 
{Classification } & {Realdisp \cite{RefWorks:207,RefWorks:358}} & {216,752 } & {117 } & {Accuracy ($\uparrow$better) } & {SVM } & {Alg. \ref{alg:RON-Gauss-for-gmm} }\tabularnewline
\hline 
{Regression } & {Twitter \cite{RefWorks:240}} & {573,820 } & {77 } & {RMSE ($\downarrow$better) } & {KRR } & {Alg. \ref{alg:RON-Gauss_supervised} }\tabularnewline
\hline 
\end{tabular}
\par\end{centering}
\caption{Summary of the experimental setups of the three experiments. }\label{tab:exp_setup_summary}
\vspace{-1em}
\end{table*}

\vspace{-1em}
\subsection{Datasets}
\vspace{-1em}
\subsubsection{Grammatical Facial Expression (GFE)}
\vspace{-1em}
This dataset is based on facial expression analysis from
video images under Libras
(a Brazilian sign language), and has 27,936 samples and 301 features \cite{RefWorks:230}. There are multiple
clusters based on different grammatical expressions. \chang{The image features are designed to be informative of the facial expressions. However, the same features may be used to infer the individuals whose images are included in the dataset. Hence, it is desirable to release a DP-protected dataset.} We use this dataset for the privacy-preserving
clustering study on Algorithm \ref{alg:unsupervised}.

\vspace{-1.5em}
\subsubsection{Realistic Sensor Displacement (Realdisp)} 
\vspace{-1em}
This is a mobile-sensing dataset used for activity recognition \cite{RefWorks:207,RefWorks:358}. The
features include readings of various motion sensors, and the goal
is to identify the activity being performed. \chang{However, the same features can possibly be used to identify the individuals whose data are in the dataset. Therefore, it is desirable to release a DP-protected dataset.} The dataset consists
of 216,752 training samples, and 1,290 testing samples with 117 features.
In our experiments, we use this dataset for the privacy-preserving
classification study with Algorithm \ref{alg:RON-Gauss-for-gmm}.
Specifically, we formulate it as a binary classification \textendash{}
identifying whether the subject is performing an action that causes
a location displacement or not, e.g. walking, running, cycling, etc.

\vspace{-1.5em}
\subsubsection{Buzz in Social Media (Twitter)}
\vspace{-1em}
This dataset extracts 77 features from Twitter posts, which are used
to predict the popularity level of the topic represented as a real
value in $[-1,1]$ \cite{RefWorks:240}. \chang{However, these features may also be used to infer the owner of each tweet; thus it is desirable to instead release the DP-protected dataset.} The dataset is divided into the training set of
573,820 samples, and the testing set of 4,715 samples. We use this
dataset for privacy-preserving regression, and adopt Algorithm \ref{alg:RON-Gauss_supervised}
for the experiments.

\vspace{-1em}
\subsection{Setups \label{subsec:Setups}}

Since RON-Gauss algorithms require $\epsilon_{\mu}$ and $\epsilon_{\Sigma}$
for the mean and the covariance, respectively, given a fixed total
privacy budget of $\epsilon$, we allocate the budget as:
$\epsilon_{\mu}=0.3\epsilon$ and $\epsilon_{\Sigma}=0.7\epsilon$.
The rationale is that the covariance is the more critical
parameter in our algorithms, and usually has higher complexity than
the mean ($\mathbb{R}^{p\times p}$ vs $\mathbb{R}^{m}$). For all
experiments, we perform 100 trials and report the average with
the 95\% confidence interval.

\vspace{-1.3em}
\subsubsection{Clustering Setup}
\vspace{-1em}
Clustering is unsupervised learning, so we apply Algorithm
\ref{alg:unsupervised}. \chang{We use K-means \cite{RefWorks:33}
as the clustering method for its simplicity and efficiency, even for large datasets, and use the \emph{Silhouette Coefficient
(S.C.)} \cite{RefWorks:232} as the metric for evaluation. The number of clusters in K-means is set using the Silhouette analysis method \cite{RefWorks:232,silhouette_analysis}.} S.C. is defined as follows. For the sample
$\mathbf{x}_{i}$ assigned to class $y(i)$, 
\begin{itemize}
\item let $a(i)$ be the average distance between $\mathbf{x}_{i}$ and
all other samples assigned to the same class $y(i)$; 
\item let $b(i)$ be the average distance between $\mathbf{x}_{i}$ and
all points assigned to the next nearest class.
\end{itemize}
\vspace{-0.5em}
Let $sc(i)=\frac{b(i)-a(i)}{\max\{b(i),a(i)\}}$, and S.C. is defined
as: 
\[
S.C.=\frac{1}{n}\sum_{i=1}^{n}sc(i).
\]
Intuitively, S.C. measures the average distance between the sample
and its class mean, normalized by the distance to the next nearest
class mean. Its range is $[-1,1]$, where higher value indicates
better the performance.

We pick this metric for two reasons. First, as opposed to other
metrics including ACC \cite{RefWorks:373}, ARI \cite{RefWorks:374},
or V-measure \cite{RefWorks:375}, S.C. does not require the knowledge
of the ground truth. This is vital both for our evaluation and for
real-world applications, respectively because the ground truth is
not available for the synthetic data in our evaluation, and it is
often not available in practice, too. Second, as suggested by Rousseeuw
\cite{RefWorks:232}, S.C. depends primarily on the distribution of
the data, but less on the clustering algorithm used, so it is fitting
for the evaluation of non-interactive private data release.

 \vspace{-1em}
\subsubsection{Classification Setup} \label{subsubsec:clf_setup}
 \vspace{-1em}
For classification, we employ the GMM according to Algorithm \ref{alg:RON-Gauss-for-gmm},
and use the support vector machine (SVM) \cite{RefWorks:359,RefWorks:231}
as the classifier in all experiments. \chang{SVM is chosen since it has been shown to perform well on binary classification \cite{RefWorks:199,RefWorks:430,byun2002applications}, and it has been proven \textendash{} both empirically and theoretically \textendash{} to generalize well \cite{RefWorks:376,RefWorks:199}. The evaluation metric is the traditional classification accuracy.
}

\chang{Since we consider the original training data as sensitive, we apply RON-Gauss to generate DP data that are used to train machine learning models. However, we test the machine learning models on the real test data in order to evaluate the ability of the DP training data to capture the classification pattern of the real data.}

\vspace{-1em}
\subsubsection{Regression Setup}
\vspace{-1em}
We use Algorithm \ref{alg:RON-Gauss_supervised} for regression, and
use kernel ridge regression (KRR) \cite{RefWorks:33,RefWorks:231}
as the regressor due to its large hypothesis class with proven theoretical error bound \cite{RefWorks:199,zhang2005learning}. The evaluation metric is the root-mean-square
error (RMSE) \cite{RefWorks:33,RefWorks:225}. Finally, for comparison, we also provide a random-guess baseline of which the prediction is drawn i.i.d. from a uniform distribution.
\chang{Finally, we manage the train/test split in a similar fashion to the above classification setup (Section \ref{subsubsec:clf_setup}).}

\vspace{-0.8em}
\subsubsection{Comparison to Other Methods}
\vspace{-0.7em}
To provide context to the experimental results, we compare our approach
to four previous works and a non-private baseline method as follows. 
\begin{enumerate}
\item Real data: the non-private baseline approach, where the result is
obtained from the original data without any modification. 
\item Li et al. \cite{RefWorks:337}: the method based on dimensionality
reduction via Bernoulli random projection on the identity query. 
\item Jiang et al. \cite{RefWorks:339}: the method based on PCA on the
identity query.
\item Blum et al. \cite{RefWorks:174}: exponential mechanism for non-interactive
setting. 
\item Liu \cite{RefWorks:372}: parametric generative model without DR. 
\end{enumerate}
We compare RON-Gauss to these five methods for the following reasons.
The first comparison is to show the real-world usability of RON-Gauss.
The second and third comparisons are to motivate the use of the Gaussian
generative model over the identity query, and the remaining comparisons
are to motivate DR via the RON projection. For all previous methods, we
use the parameters suggested by the respective authors, and we vary
the hyper-parameter before reporting the best result.

\vspace{-1em}
\subsection{Experimental Results}
\vspace{-0.5em}
For methods with DR, the results reported are the best results among 
varied dimensions.\footnote{As discussed by Chaudhuri et al. \cite{RefWorks:188}, in the real-world
deployment, the parameter tuning process must be private as well.}

\begin{table}
\begin{centering}
\begin{tabular}{>{\centering}p{1.7cm}>{\centering}p{1.1cm}>{\centering}p{1cm}>{\centering}p{0.25cm}>{\centering}p{1.3cm}>{\centering}p{0.9cm}}
\toprule 
{Method }  & {Model }  & {DR}  & {$\epsilon$}  & {S.C.}  & {$\Delta$S.C.}\tabularnewline
\midrule
\midrule 
{Real data }  & {$-$ }  & {$-$ }  & {$-$ }  & {$.286$}  & {$.00$}\tabularnewline
\midrule 
{Li et al. \cite{RefWorks:337} }  & {Identity }  & {Bern. Rand. }  & {1.}  & {$.123\pm.000$}  & {$.16$}\tabularnewline
\midrule 
{Jiang et al. \cite{RefWorks:339} }  & {Identity }  & {PCA }  & {1.}  & {$.123\pm.000$}  & {$.16$}\tabularnewline
\midrule 
{Blum et al. \cite{RefWorks:174} }  & {Exp. Mech. }  & {$-$ }  & {1.}  & {$.026\pm.017$}  & {$.26$}\tabularnewline
\midrule 
{Liu \cite{RefWorks:372} }  & {Gaussian }  & {$-$ }  & {1.}  & {$.092\pm.001$}  & {$.19$}\tabularnewline
\midrule 
{RON-Gauss }  & {Gaussian }  & {RON }  & {1.}  & {$.274\pm.015$}  & {$.01$}\tabularnewline
\bottomrule
\end{tabular}
\par\end{centering}
\caption{Clustering results (GFE dataset). $\Delta$S.C.
indicates the error relative to the performance by real data. \label{tab:Clustering-results}}
\vspace{-1em}
\end{table}

\vspace{-1.5em}
\subsubsection{Privacy-Preserving Clustering}
\vspace{-1em}
Table \ref{tab:Clustering-results} summarizes the results for clustering,
and the following are main observations. 
\begin{itemize}
\item Compared to the non-private baseline (real data), RON-Gauss has almost
identical performance with only 0.01 additional error (4\% error)
while preserving strong privacy ($\epsilon=1.0$). 
\item Compared to Li et al. \cite{RefWorks:337} and Jiang et al. \cite{RefWorks:339},
who use the identity query as opposed to the Gaussian generative model,
RON-Gauss has over 2x better utility with the same privacy budget. 
\item Compared to Blum et al. \cite{RefWorks:174} and Liu \cite{RefWorks:372},
who do not use DR, RON-Gauss has over 10x and 3x better utility with
the same privacy budget. 
\end{itemize}

\chang{For RON-Gauss, the optimal number of clusters based on the Silhouette analysis is four. It is  interesting to note that RON-Gauss achieves good results despite using the unimodal Gaussian model. This can partially be explained by the curse of dimensionality
~\cite{RefWorks:350,beyer1999nearest,koppen2000curse,donoho2000high}. One consequence of the curse of dimensionality is the concentration of the data mass near the surface of the hypercube encapsulating the data domain space. With respect to our results, this leads to the observation that despite using the unimodal Gaussian model the data generated by RON-Gauss can form different clusters around different parts of the hypercube surface. 
Thus, if this unimodal distribution can represent the original data well according to the DFM effect, it can provide clustering performance close to that of the original data. 
}

\vspace{-2em}
\subsubsection{Privacy-Preserving Classification}
\vspace{-1em}
Table \ref{tab:The-classification-results} summarizes the classification
results. The following are main observations. 
\begin{itemize}
\item Compared to the non-private baseline (real data), RON-Gauss has almost
identical performance with 2.45\% additional error, while preserving
strong privacy ($\epsilon=1.0$). 
\item Compared to Li et al. \cite{RefWorks:337} and Jiang et al. \cite{RefWorks:339},
who use the identity query as opposed to GMM, RON-Gauss has over 20\%
and 30\% better utility, respectively, with the same privacy budget. 
\item Compared to Blum et al. \cite{RefWorks:174} and Liu \cite{RefWorks:372},
who do not use DR, RON-Gauss has over 35\% and 25\% better utility,
respectively, with the same privacy. 
\end{itemize}
\vspace{-2em}

\begin{table}
\begin{centering}
\begin{tabular}{>{\centering}p{1.95cm}>{\centering}p{1.3cm}>{\centering}p{1.2cm}cc}
\toprule 
{Method }  & {Model }  & {DR}  & {$\epsilon$}  & {Accuracy (\%)}\tabularnewline
\midrule
\midrule 
{Real data }  & {- }  & {- }  & {- }  & {$89.61$}\tabularnewline
\midrule 
{Li et al. \cite{RefWorks:337} }  & {Identity }  & {Bern. Rand. }  & {1. }  & {$65.04\pm0.90$}\tabularnewline
\midrule 
{Jiang et al. \cite{RefWorks:339} }  & {Identity }  & {PCA} & {1. }  & {$54.51\pm1.65$}\tabularnewline
\midrule 
{Blum et al. \cite{RefWorks:174} }  & {Exp. Mech. }  & {- }  & {1. }  & {$51.24\pm1.83$}\tabularnewline
\midrule 
{Liu \cite{RefWorks:372} }  & {GMM }  & {- }  & {1. }  & {$61.31\pm0.65$}\tabularnewline
\midrule 
{RON-Gauss}  & {GMM }  & {RON}  & {1.}  & {$87.16\pm0.27$}\tabularnewline
\bottomrule
\end{tabular}
\par\end{centering}
\caption{Classification results (Realdisp dataset).\label{tab:The-classification-results}}
\vspace{-1em}
\end{table}

\vspace{-0.7em}
\subsubsection{Privacy-Preserving Regression}
\vspace{-1em}
Table \ref{tab:Regression-results} summarizes the results for regression.
The following are main observations. 
\begin{itemize}
\item Compared to the non-private baseline (real data), RON-Gauss actually
performs statistically equally well, while preserving strong privacy
($\epsilon=1.0$). 
\item Compared to Li et al. \cite{RefWorks:337} and Jiang et al. \cite{RefWorks:339},
who use the identity query as opposed to the Gaussian generative model,
RON-Gauss has over 3x better utility with the same privacy budget. 
\item Compared to Blum et al. \cite{RefWorks:174} and Liu \cite{RefWorks:372},
who do not use DR, RON-Gauss has over 3x and 5x better utility with
the same privacy budget.
\end{itemize}

\vspace{-2em}
\subsubsection{Summary of Experimental Results}
\vspace{-1em}
RON-Gauss outperforms all four other methods in terms of
utility across all three learning tasks. RON-Gauss also performs comparably
well relative to the maximum utility achieved by the non-private baseline
in all tasks. The main results are concluded as follows. 
\begin{itemize}
\item RON-Gauss provides performance close to that attainable from the non-private
real data. 
\item Using the Gaussian generative model over the identity query has been
shown to provide the utility gain of up to 2x, 30\%, and 3x for clustering,
classification, and regression, respectively. 
\item Using RON to reduce dimension of the data has been shown to provide
the utility gain of up to 10x, 35\%, and 5x, for clustering, classification,
and regression, respectively. 
\end{itemize}

\vspace{-2em}

\begin{table}
\begin{centering}
\begin{tabular}{>{\centering}p{1.95cm}>{\centering}p{1.3cm}>{\centering}p{1cm}cc}
\toprule 
{Method }  & {Model }  & {DR} & {$\epsilon$}  & {RMSE ($\times10^{-2})$}\tabularnewline
\midrule
\midrule 
{Real data }  & {- }  & {- }  & {- }  & {$0.21$}\tabularnewline
\midrule 
{Li et al. \cite{RefWorks:337} }  & {Identity }  & {Bern. Rand. }  & {1. }  & {$0.68\pm0.01$}\tabularnewline
\midrule 
{Jiang et al. \cite{RefWorks:339} }  & {Identity }  & {PCA }  & {1. }  & {$0.68\pm0.00$}\tabularnewline
\midrule 
{Blum et al. \cite{RefWorks:174} }  & {Exp. Mech. }  & {- }  & {1. }  & {$0.62\pm0.07$}\tabularnewline
\midrule 
{Liu \cite{RefWorks:372} }  & {Gaussian}  & {- }  & {1.}  & {$1.00\pm0.12$}\tabularnewline
\midrule 
{RON-Gauss }  & {Gaussian}  & {RON} & {1. }  & {$0.21\pm0.01$}\tabularnewline
\bottomrule
\end{tabular}
\par\end{centering}
\caption{Regression results (Twitter dataset). RMSE is an error metric, so
lower values indicate better utility. (note: RMSE of random guess
is $\sim57.20\times10^{-2}$).\label{tab:Regression-results}}
\vspace{-1em}
\end{table}

\section{Discussion} \label{sec:discussion}

\subsection{Effect of Dimension on the Utility}

\begin{figure}
\begin{centering}
\includegraphics[scale=0.5]{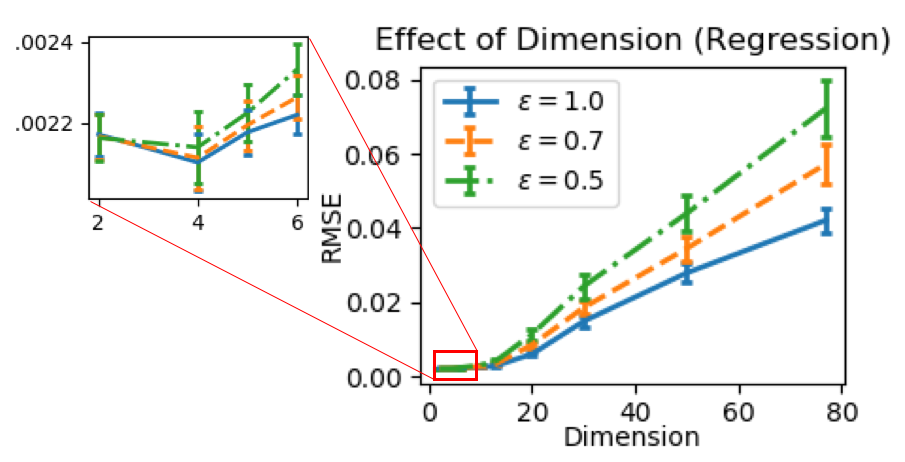} 
\par\end{centering}
\caption{Effects of dimension on the regression performance on Twitter dataset.
\label{fig:Effects-of-dimension}}
\end{figure}

In Section \ref{sec:RON-Gauss}, we discuss how RON projection
can reduce the level of noise required for DP. This effect can be
observed experimentally as illustrated by Figure \ref{fig:Effects-of-dimension},
which shows the relationship between the dimension the data are reduced
to and the utility performance. Noticeably, there is a gain in utility
as the dimension is reduced. Specifically, the peak performance is
achieved at 4 dimensions in this case. This general trend is consistent
across different privacy budgets. Seeking the optimal dimension a
priori is an interesting topic for future research on the RON-projection-based
methods.

\vspace{-1em}
\subsection{RON-Gauss Against Membership Inference Attack}

\begin{figure}
\begin{centering}
\includegraphics[scale=0.55]{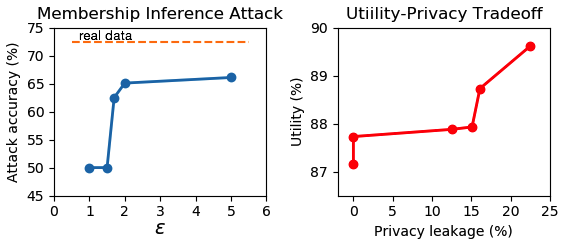} 
\par\end{centering}
\caption{\chang{Membership inference attack on RON-Gauss using Realdisp dataset. (Left) The attack accuracy against different values of $\epsilon$. The dashed line shows the attack accuracy on the real data for comparison. (Right) The tradeoff between the utility (classification accuracy) and the privacy leakage (the difference between the membership inference accuracy and random guess at 50\%)}.
\label{fig:inference-attack}}
\end{figure}

\chang{Recent works have suggested using inference attacks to measure the susceptibility of the released data and identify the appropriate values of $\epsilon$ for non-interactive differential privacy, e.g. \cite{RefWorks:542,balu2014challenging}. To evaluate RON-Gauss against inference attacks, we implement the membership inference attack proposed by Shokri et al. \cite{RefWorks:542} using their published software \cite{membership-att}. This attack trains shadow machine learning models and an attack model to identify whether a given sample is in the dataset. Since their attack is designed for a classification task, we evaluate it on our classification experimental setup using Realdisp data. For the attack, we use ten shadow models and use neural network for the attack model with 0.01 learning rate trained on 50 epochs. These are the default parameter values of the software used \cite{membership-att}.}

\chang{The results are shown in Figure \ref{fig:inference-attack}. The test set is chosen such that a random guess on the membership inference attack would yield an accuracy of 50\%. Figure \ref{fig:inference-attack} (Left) suggests that $\epsilon$ values of 1.5 or less are appropriate for this setting since the performance of the membership inference attack is close to a random guess. Figure \ref{fig:inference-attack} (Right) illustrates the utility-privacy tradeoff based on this attack. The privacy leakage is defined as the attack accuracy above the random guess level. In other words, it measures how much the attack performs better than a random guess. The utility is measured by the classification accuracy similar to our classification experiments in Section \ref{sec:Experiments}. The plot allows the practitioners to choose an $\epsilon$ value that meets their utility-privacy tradeoff. For example, if we require the privacy leakage to be less than 10\%, the curve in Figure \ref{fig:inference-attack} (Right) suggests that we can achieve almost 88\% utility.
}

\chang{We note that the membership inference attack of Shokri et al.~\cite{RefWorks:542} may not be the optimal inference attack against RON-Gauss, since it is a general attack method not specifically tailored for our approach. We leave the analysis of more advanced attacks that specifically utilize knowledge of the RON-Gauss mechanism to future work,   e.g. using hypothesis testing~\cite{balu2014challenging}.}

\vspace{-1em}
\subsection{RON-Gauss as a Generative Model}

Our work uses a parametric generative model to capture the essence of the unknown data distribution. Since RON-Gauss involves DR as an important step, the RON-Gauss model is inevitably lossy, i.e. there is information loss due to the use of the model itself. However, this loss is mitigated partly by the DFM effect, which ensures that the data are close to Gaussian after the RON projection. To illustrate the effectiveness of this effect and of RON-Gauss as a parametric generative model, we test RON-Gauss purely for its quality as a generative model, i.e. without the DP component, on the MNIST dataset \cite{RefWorks:541,lecun-mnist}. 
Since MNIST is typically used for classification, we use Algorithm \ref{alg:RON-Gauss-for-gmm} for RON-Gauss and set $\epsilon\rightarrow\infty$ to leave out the effect of DP noise. We project the data onto 392 dimensions \textendash{} half of the original dimensions of 784 \textendash{} and synthesize the samples, which are then reconstructed into the synthesized images. Examples of the synthesized images are shown in Figure \ref{fig:mnist_ex}. These images show good digit visibility, which indicates the potential of RON-Gauss as a generative model. 

However, admittedly, the visibility of the digits subsides gradually as we project the data onto lower dimensions. Particularly, we observe that, at dimensions lower than 100, the digits are not very visible anymore. This depicts that, despite its promise, RON-Gauss may not yet be the universal model for every situation since there remains the need to balance the information loss due to DR. However, if sufficient information is retained, RON-Gauss has shown the potential to be a quality model by utilizing the DFM effect,as demonstrated by our experiments in Section \ref{sec:Experiments}.

\begin{figure}
\begin{centering}
\includegraphics[scale=0.45]{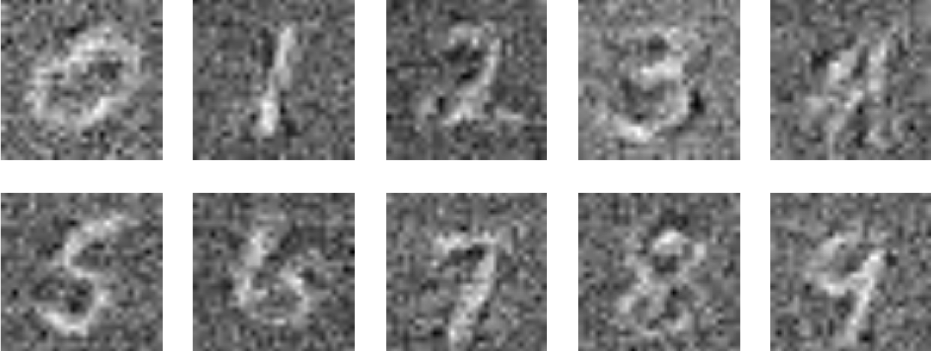} 
\par\end{centering}
\caption{Synthesized MNIST images from the RON-Gauss model without the DP component using half of the full dimensions.
\label{fig:mnist_ex}
\vspace{-0.5em}}
\end{figure}

\vspace{-1em}
\subsection{The Design of \texorpdfstring{$\epsilon_{\mu}$}{eps\_mu} and \texorpdfstring{$\epsilon_{\Sigma}$}{eps\_Sigma} for RON-Gauss Algorithms}

RON-Gauss Algorithms take as inputs two privacy parameters: $\epsilon_{\mu}$ and $\epsilon_{\Sigma}$. The algorithms are then shown to preserve $(\epsilon_{\mu}+\epsilon_{\Sigma})$-differential privacy. This means that, for a fixed total privacy budget of $\epsilon=\epsilon_{\mu}+\epsilon_{\Sigma}$, we can choose how much to allocate to $\epsilon_{\mu}$ and $\epsilon_{\Sigma}$ in order to maximize the utility of the synthesized data. In our experiments, we fix the ratio between the two based on the observation about the sensitivity of the mean and the covariance. However, the allocation can possibly be designed better by formulating it as an optimization problem that aims at maximizing the utility of the synthetic data. Then, the optimal solution can be obtained using grid search, random search, or Bayesian optimization \cite{snoek2012practical}. We leave this as a possible future direction.

\vspace{-1em}
\section{Conclusion}
\vspace{-0.5em}
In this work, we combine two previously non-intersecting techniques
\textendash{} random orthonormal projection and Gaussian generative
model \textendash{} to provide a solution to non-interactive private data release. 
We propose the RON-Gauss model that exploits the Diaconis-Freedman-Meckes effect, and
present three algorithms for both unsupervised and supervised learning. We prove that our RON-Gauss model preserves $\epsilon$-differential
privacy. Finally, our experiments on three real-world datasets under
clustering, classification, and regression applications show the strength
of the method. RON-Gauss provides significant performance
improvement over previous approaches, and yields
the utility performance close to the non-private baseline, while preserving 
differential privacy with $\epsilon=1$.

\section*{Acknowledgement}
\chang{
The authors would like to thank S{\'e}bastien Gambs for shepherding the paper, the anonymous reviewers for their valuable feedback, and Mert Al, Daniel Cullina, and Alex Dytso for  insightful discussions. This work is supported in part by the National Science Foundation (NSF) under the grant CNS-1553437 and CCF-1617286, an Army Research Office YIP Award, and faculty research awards from Google, Cisco, Intel, and IBM.}

{\bibliographystyle{plain}
\bibliography{dimdp_references}

\appendix

\section{Proof of Lemma \ref{lem:proj_range}}
\label{sec:proof_lem_proj_range}

\begin{proof}
The proof uses the property of orthogonal projection in a vector space.
First, notice that $\left\Vert \mathbf{W}^{T}\mathbf{x}\right\Vert _{F}=\left\Vert \mathbf{W}\mathbf{W}^{T}\mathbf{x}\right\Vert _{F}$,
which can be verified as follows.
\begin{align*}
\left\Vert \mathbf{W}^{T}\mathbf{x}\right\Vert _{F} & =\sqrt{\mathrm{tr}(\mathbf{x}^{T}\mathbf{W}\mathbf{W}^{T}\mathbf{x})}\\
 & =\sqrt{\mathrm{tr}(\mathbf{x}^{T}\mathbf{W}\mathbf{W}^{T}\mathbf{W}\mathbf{W}^{T}\mathbf{x})}\\
 & =\left\Vert \mathbf{W}\mathbf{W}^{T}\mathbf{x}\right\Vert _{F},
\end{align*}
where the second equality is from the fact that $\mathbf{W}^{T}\mathbf{W}=\mathbf{I}$.
Then, notice that $\mathbf{P}=\mathbf{W}\mathbf{W}^{T}$ is a projection
matrix with $p$ orthonormal basis as the columns of $\mathbf{W}$
(cf. \cite[Chapter 5]{RefWorks:290}). Therefore, the idempotent property
of $\mathbf{P}$ can be used as,
\begin{align*}
\left\Vert \mathbf{W}^{T}\mathbf{x}\right\Vert _{F} & =\left\Vert \mathbf{W}\mathbf{W}^{T}\mathbf{x}\right\Vert _{F}=\left\Vert \mathbf{P}\mathbf{x}\right\Vert _{F}\\
 & =\left\langle \mathbf{P}\mathbf{x},\mathbf{P}\mathbf{x}\right\rangle _{F}=\left\langle \mathbf{P}\mathbf{x},\mathbf{x}\right\rangle _{F}.
\end{align*}
The last equality can be verified as follows. Let $\mathbf{P}\mathbf{x}\in\mathcal{P}$,
and let $\mathbf{x}^{\bot}=\mathbf{x}-\mathbf{P}\mathbf{x}\in\mathcal{P}^{\bot}$,
then $\left\langle \mathbf{P}\mathbf{x},\mathbf{x}\right\rangle _{F}=\left\langle \mathbf{P}\mathbf{x},\mathbf{P}\mathbf{x}+\mathbf{x}^{\bot}\right\rangle _{F}=\left\langle \mathbf{P}\mathbf{x},\mathbf{P}\mathbf{x}\right\rangle _{F}+\left\langle \mathbf{P}\mathbf{x},\mathbf{x}^{\bot}\right\rangle _{F}=\left\langle \mathbf{P}\mathbf{x},\mathbf{P}\mathbf{x}\right\rangle _{F}$
from the additivity of the inner product and the fact that $\left\langle \mathbf{P}\mathbf{x},\mathbf{x}^{\bot}\right\rangle _{F}=0$
by construction. Then, using the Cauchy-Schwarz inequality, $\left\Vert \mathbf{P}\mathbf{x}\right\Vert _{F}^{2}=\left\langle \mathbf{P}\mathbf{x},\mathbf{x}\right\rangle _{F}^{2}\leq\left\Vert \mathbf{P}\mathbf{x}\right\Vert _{F}\left\Vert \mathbf{x}\right\Vert _{F}$,
and, hence, $\left\Vert \mathbf{P}\mathbf{x}\right\Vert _{F}=\left\Vert \mathbf{W}\mathbf{W}^{T}\mathbf{x}\right\Vert _{F}=\left\Vert \mathbf{W}^{T}\mathbf{x}\right\Vert _{F}\leq\left\Vert \mathbf{x}\right\Vert _{F}$.
\end{proof}

\section{\texorpdfstring{$L_{1}$}{L1}-Sensitivity of the MLE of the Covariance}

\label{sec:sens_mle_cov}

Consider the maximum likelihood estimate (MLE) \cite{RefWorks:532}
for the covariance matrix, which is an unbiased estimate (cf. \cite{RefWorks:504}):
\[
\boldsymbol{\Sigma}_{MLE}=\frac{1}{n}\sum_{i=1}^{n}(\widetilde{\mathbf{x}_{i}}-\boldsymbol{\mu})(\widetilde{\mathbf{x}_{i}}-\boldsymbol{\mu})^{T},
\]
where $\boldsymbol{\mu}$ is the sample mean specific to the instance
of the dataset. Hence, the neighboring datasets may have different
means. Then, the sensitivity can be derived as follows.
\begin{lem}
The $L_{1}$-sensitivity of the MLE of the covariance matrix is $(2\sqrt{p}+2n\sqrt{p})/n$.
\end{lem}
\begin{proof}
For neighboring datasets $\mathbf{X},\mathbf{X}'$,
\begin{align*}
s(f)= & \sup\frac{1}{n}\left\Vert \sum_{i=1}^{n}(\widetilde{\mathbf{x}_{i}}-\boldsymbol{\mu})(\widetilde{\mathbf{x}_{i}}-\boldsymbol{\mu})^{T}\right.\\
 & \left.-\sum_{i=1}^{n}(\widetilde{\mathbf{x}'_{i}}-\boldsymbol{\mu}')(\widetilde{\mathbf{x}'_{i}}-\boldsymbol{\mu}')^{T}\right\Vert _{1}\\
 = & \sup\frac{1}{n}\left\Vert (\sum_{i=1}^{n}\widetilde{\mathbf{x}_{i}}\widetilde{\mathbf{x}_{i}}^{T}-n\boldsymbol{\mu}\boldsymbol{\mu}^{T})\right.\\
 & \left.-(\sum_{i=1}^{n}\widetilde{\mathbf{x}'_{i}}\widetilde{\mathbf{x}'_{i}}^{T}-n\boldsymbol{\mu}'\boldsymbol{\mu}'^{T})\right\Vert _{1}\\
= & \sup\frac{1}{n}\left\Vert (\widetilde{\mathbf{x}_{i}}\widetilde{\mathbf{x}_{i}}^{T}-\widetilde{\mathbf{x}_{i}'}\widetilde{\mathbf{x}_{i}'}^{T})+n(\boldsymbol{\mu}'\boldsymbol{\mu}'^{T}-\boldsymbol{\mu}\boldsymbol{\mu}^{T})\right\Vert _{1}\\
\leq & \sup\frac{1}{n}(\left\Vert \widetilde{\mathbf{x}_{i}}\widetilde{\mathbf{x}_{i}}^{T}\right\Vert _{1}+\left\Vert \widetilde{\mathbf{x}_{i}'}\widetilde{\mathbf{x}_{i}'}^{T}\right\Vert _{1})\\
 & +\left\Vert \boldsymbol{\mu}'\boldsymbol{\mu}'^{T}\right\Vert _{1}+\left\Vert \boldsymbol{\mu}\boldsymbol{\mu}^{T}\right\Vert _{1}\\
\leq & \sup\frac{2\sqrt{p}}{n}\left\Vert \widetilde{\mathbf{x}_{i}}\widetilde{\mathbf{x}_{i}}^{T}\right\Vert _{F}+2\sqrt{p}\left\Vert \boldsymbol{\mu}'\boldsymbol{\mu}'^{T}\right\Vert _{F}\\
 \leq & \frac{2\sqrt{p}}{n}+2\sqrt{p}.
\end{align*}
The last inequality is due to the following observation: $\left\Vert \boldsymbol{\mu}\right\Vert _{F}=\frac{1}{n}\left\Vert \sum\widetilde{\mathbf{x}_{i}}\right\Vert _{F}\leq\frac{1}{n}\sum\left\Vert \widetilde{\mathbf{x}_{i}}\right\Vert _{F}=1$.
\end{proof}

\end{document}